\documentclass[a4paper,USenglish]{lipics-v2021}
\usepackage{graphicx} % Required for inserting images
\usepackage{multirow}
\title{Temporal Cycle Detection and Acyclic Temporization}

\author{Davi de Andrade}{Universidade Federal do Ceará, Brazil}{daviandradeiacono@gmail.com}{}{}

\author{Júlio Araújo}{Universidade Federal do Ceará, Brazil}{julio@mat.ufc.br}{}{}

\author{Allen Ibiapina}{IRIF, CNRS \& Université Paris Cité, France}{Allen-Roossim.Passos-Ibiapina@irif.fr}{}{Supported by the French ANR, project ANR-22-CE48-0001 (TEMPOGRAL)}

\author{Andrea Marino
}{Universita degli Studi di Firenze, Italy}{andrea.marino@unifi.it}{}{Partially funded by Italian PNRR CN4 Centro Nazionale per la Mobilità Sostenibile, NextGeneration EU - CUP,  B13C22001000001. MUR of Italy, under PRIN Project n. 2022ME9Z78 - NextGRAAL: Next-generation algorithms for constrained GRAph visuALization, PRIN PNRR Project n. P2022NZPJA - DLT-FRUIT: A user centered framework for facilitating DLTs FRUITion.}

\author{Jason Schoeters
}{Universita degli Studi di Firenze, Italy}{jason.schoeters@unifi.it}{}{Supported by the French ANR, project ANR-22-CE48-0001 (TEMPOGRAL), and partially funded by Italian PNRR CN4 Centro Nazionale per la Mobilità Sostenibile, NextGeneration EU - CUP,  B13C22001000001. MUR of Italy, under PRIN PNRR Project n. P2022NZPJA - DLT-FRUIT: A user centered framework for facilitating DLTs FRUITion.}

\author{Ana Silva
}{Universidade Federal do Ceará, Brazil}{anasilva@mat.ufc.br}{}{}

\date{}
\usepackage{tikz-network}
\usepackage{pgf,tikz}

\usetikzlibrary{decorations}
\usetikzlibrary{decorations.markings, quotes}

\usepackage{complexity}

\usepackage{mathrsfs}
\usepackage{amsfonts}

\usepackage{caption}
\usepackage{tabularx}
\usepackage{makecell}
\usepackage{mdframed}
\usepackage{arydshln}
\usepackage{tabu}
\usepackage{subcaption}
\usepackage{xcolor}
\usepackage{amsmath}
\usepackage{amssymb}
\usepackage{thmtools} 
\usepackage{thm-restate}

\usepackage{float}

\usepackage[ruled, noend, linesnumbered]{algorithm2e}
\SetKwIF{If}{ElseIf}{Else}{if}{:}{else if}{else :}{end if}%
\SetKwFor{For}{for}{:}{}%

\SetCommentSty{mycommfont}

\definecolor{ForestGreen}{RGB}{34,139,34}

\definecolor{corsimcornao}{RGB}{153, 0, 153}

\AtBeginEnvironment{algorithm}{\SetArgSty{textrm}}
\SetKwComment{Comment}{/* }{ */}
\SetKwFor{For}{for }{}{end}
\SetKwIF{If}{ElseIf}{Else}{if}{}{elif}{else}{}%
\usepackage{cleveref}

\usepackage{soul}

\newcommand{\tG}{\mathcal{D}\xspace}
\newcommand{\tpath}[1]{temporal #1-path\xspace}
\newcommand{\simplecycle}{simple-cycle\xspace}
\newcommand{\weakcycle}{weak-cycle\xspace}
\newcommand{\strongcycle}{strong-cycle\xspace}
\newcommand{\simplecycles}{simple-cycles\xspace}
\newcommand{\weakcycles}{weak-cycles\xspace}
\newcommand{\strongcycles}{strong-cycles\xspace}

\newcommand{\CycleDetection}{\textsc{Cycle Detection}\xspace}
\newcommand{\AcyclicTemporization}{\textsc{Acyclic Temporization}\xspace}
\newcommand{\ThreeSAT}{\textsc{3-SAT}\xspace}
\newcommand{\lexicographicTemporization}{lexicographic temporization\xspace}



\Copyright{Davi de Andrade, Júlio Araújo, Allen Ibiapina, Andrea Marino, Jason Schoeters, and Ana Silva} %TODO mandatory, please use full first names. LIPIcs license is "CC-BY";  http://creativecommons.org/licenses/by/3.0/

\authorrunning{
	D. de Andrade, J. Araújo, A. Ibiapina,
	A. Marino,
	J. Schoeters,
	A. Silva}

\nolinenumbers

\begin{document}

\maketitle

\begin{abstract}
    In directed graphs, a cycle can be seen as a structure that allows its vertices to loop back to themselves, or as a structure that allows pairs of vertices to reach each other through distinct paths. We extend these concepts to temporal graph theory, resulting in multiple interesting definitions of a ``temporal cycle". For each of these, we consider the problems of \textsc{Cycle Detection} and \textsc{Acyclic Temporization}. For the former, we are given an input temporal digraph, and we want to decide whether it contains a temporal cycle. Regarding the latter, for a given input (static) digraph, we want to time the arcs such that no temporal cycle exists in the resulting temporal digraph. We're also interested in \textsc{Acyclic Temporization} where we bound the lifetime of the resulting temporal digraph. Multiple results are presented, including polynomial and fixed parameter tractable search algorithms, polynomial-time reductions from \textsc{3-SAT} and \textsc{Not All Equal 3-SAT}, and temporizations resulting from arbitrary vertex orderings which cover (almost) all cases. 
\end{abstract}

\newpage
\setcounter{page}{1}

\section{Introduction}

%\textbf{Temporal digraphs.}
A \emph{temporal digraph  with lifetime $\tau$} is a pair $\mathcal{D} = (D,\lambda)$ where  $D$ is a directed graph (or digraph), called \emph{underlying} digraph, and $\lambda$ is a function from $A(G)$ to $2^{[\tau]}$, called \emph{time function}, or \emph{temporization}. 
Temporal graphs are powerful for analyzing dynamic relationships and patterns over time. They are widely applied in social networks (e.g., trend detection, influencer analysis), epidemiology (disease spread modeling), and transportation (route optimization), and in general in contexts where evolving connections are key to understanding behavior, predicting events, and optimizing performance~\cite{M.16,Netal.13,Holme.15,LVM.18}.

In temporal digraphs, a path\footnote{All paths are considered to be directed paths.} from a vertex $x$ to a vertex $y$, called a \emph{\tpath{$x,y$}}, is meaningful only if the times on its arcs follow a strictly increasing or non-decreasing 
%non-strictly increasing
sequence. The former, known as \emph{strict model}, %applies when transitions must occur at subsequent times. F
is applied, for example, in the representation of a public transportation system where each arc in the path corresponds to a bus or train that must be taken at a time which is later than the previous transport. The latter, called \emph{non-strict model}, allows transitions to occur instantaneously. %at the same or successive times. 
This is useful for scenarios like daily route availability, where multiple routes may be traversed in the same day if they are accessible.

\textbf{Cycles in static digraphs.} In static digraphs, a cycle is a simple non-trivial path 
%\julio{usually path definition does not allow repetition of vertices, and symmetric arcs or loops are considered to be cycles.} 
that starts and finishes at the same vertex.
%, with no other repetitions of vertices or edges along the way. 
Cycles are fundamental structures within digraphs and they come up in a wide variety of applications, from computer science to engineering, biology, and social network analysis. For example, cycles are important in network routing to avoid routing loops and enhance efficiency. In operating systems and databases, deadlocks can be represented as cycles in a resource-allocation graph. In  biochemical networks and protein interaction networks, cycles can represent feedback loops or recurring processes.  The study of cycles, cycle detection, and cycle characterization is therefore central to graph theory and its applications in the real world.
%The following \textbf{trivial properties} of cycles in static digraphs are true:
The following \textbf{fundamental properties} of cycles in static digraphs trivially hold and are equivalent in the static context:
%\emph{(i)} For every (resp. there exists a) vertex $x$ in the cycle, $x$ is able to traverse the cycle and go back to itself. \emph{(ii)} For every (resp. there exists a) pair of vertices $x,y$ in the cycle, $x$ is able to reach $y$ and $y$ is able to reach $x$ using the arcs involved in the cycle.
\emph{(i)} there exists a vertex $x$ in the cycle such that from $x$ we can traverse the cycle and go back to $x$; \emph{(ii)} there exists a pair of vertices $x,y$ in the cycle, such that $x$ is able to reach $y$ and $y$ is able to reach $x$ using the arcs involved in the cycle; \emph{(iii)} for every vertex $x$ in the cycle, starting from $x$, we can traverse the cycle and go back to $x$; and \emph{(iv)} for every pair of vertices $x,y$ in the cycle, $x$ is able to reach $y$ and $y$ is able to reach $x$ using the arcs involved in the cycle. As said, for static digraphs, all of the four statements are equivalent and, note that \emph{(iii)} is the \emph{for every} version of \emph{(i)} and \emph{(iv)} is the \emph{for every} version of \emph{(ii)}.

\textbf{Cycle Definitions in Temporal graphs.}
Inspired by the properties \emph{(i)}-\emph{(iv)} above, we can define cycles in temporal digraphs, looking for cycles in the underlying digraph whose times satisfy such properties. Interestingly, while these properties are equivalent in static digraphs, in temporal digraphs they differ, and it makes sense to study both the \emph{for every} and the \emph{there exists} variations.

We define the following types of cycles, considering (now and in the rest of the paper) only non-trivial temporal paths. In particular, given a temporal digraph $\tG = (D,\lambda)$ and a cycle $C$ of $D$,  we say that $C$ is a temporal:% \emph{temporal cycle} of type \simplecycle, \weakcycle, \strongcycle\ if:

\begin{description}
    \item[\simplecycle] if there exists a \tpath{$x,x$} $P$ such that $E(P) = E(C)$, for some $x\in V(C)$;
    \item[\weakcycle] if there exist a \tpath{$x,y$} $P$ and a \tpath{$y,x$} $P'$ such that $E(P) \cup E(P') = E(C)$, for some pair $x,y\in V(C)$;
    \item[\strongcycle] if there exists a \tpath{$x,x$} $P$ such that $E(P) = E(C)$, for every $x\in V(C)$.
\end{description}

\begin{figure}[t]
\centering
\scalebox{0.9}{
\begin{subfigure}[b]{0.3\textwidth}
\centering\includegraphics[height=2.5cm]{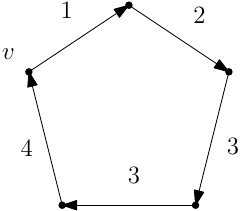}
\caption{}
\end{subfigure}%
\hfill
\begin{subfigure}[b]{0.3\textwidth}
\centering\includegraphics[height=2.5cm]{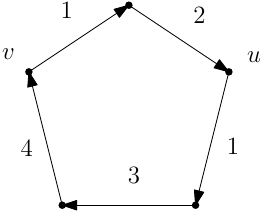}
\caption{}
\end{subfigure}%
\hfill
\begin{subfigure}[b]{0.3\textwidth}
\centering\includegraphics[height=2.5cm]{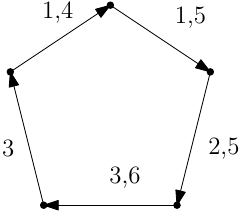}
\caption{}
\end{subfigure}%
}
\caption{\label{fig:cycles}Examples of a \simplecycle, \weakcycle, and \strongcycle respectively, using the non-strict model.\vspace{-0.3cm}}
\end{figure}

See Figure~\ref{fig:cycles} for an example using the non-strict model. Figure~\ref{fig:cycles}a corresponds to a \simplecycle as there is a vertex, namely $v$, able to go back to itself. 
Figure~\ref{fig:cycles}b corresponds to a \weakcycle, as there is a pair of vertices, namely $v$ and $u$, such that $v$ is able to go to $u$ and $u$ is able to go to $v$. Note that this is not a \simplecycle because the two paths do not compose to allow $v$ (or any other vertex) to go back to itself. Finally, Figure~\ref{fig:cycles}c corresponds to a \strongcycle as every vertex is able to go back to itself.

Note that the definitions of simple, weak, and strong cycles correspond to properties \emph{(i)}, \emph{(ii)}, and \emph{(iii)}, respectively. We miss only property \emph{(iv)}, which can be regarded as the \emph{for every} version of a weak cycle. However, it is easy to prove that this in fact would be equivalent to a \strongcycle. 
Note additionally that every \strongcycle is also a \simplecycle, as the former is the \emph{for every} version of the latter, and that every {\simplecycle} $C$ is a {\weakcycle} (just consider $y=x$).

\textbf{Problem definition.}
The first problem that arises naturally, is the detection of our temporal cycles, as described next where Type \textsc{T} can be \simplecycle, \weakcycle, or \strongcycle. \\%For each type of temporal cycle, being \weakcycle, \simplecycle, and \strongcycle, we thus consider the problem in which we are given a temporal graph, and we have to search for a temporal cycle.\\
\begin{minipage}{\textwidth}
\begin{mdframed}
    \underline{{T} \CycleDetection}\smallskip\\
\textbf{Input:} Temporal digraph $\tG$ \smallskip\\
\textbf{Question:} Does $\tG$ contain a temporal cycle of Type \textsc{T}?
\end{mdframed} 
\end{minipage}

We are interested in the time complexity of such problems. %of any size. 
Note that we do not require the temporal cycle to be of at least some given size $k$, since this would be trivially \NP-complete, by reducing from \textsc{Hamiltonian Cycle}.
%Our algorithms presented in \Cref{sec:cycledetection} are similar in this aspect as we also use or adapt a search algorithm.

Detecting a temporal cycle can also be seen as recognizing whether the given temporal graph is acyclic, which relates to the problem of recognizing DAGs\footnote{Used short for Directed Acyclic Graph.} in the static context, which can be done easily through a search algorithm. Observe that constructing a DAG $D$ from a given graph $G$, i.e., orienting the edges of $G$ so that $D$ does not contain any cycle, can be trivially done by picking an ordering of $V(G)$ and orienting all edges from smaller to bigger vertices. 
When adapted to the temporal context, such orientation would clearly still work, but what happens when the digraph is already known and, instead, we want to find a time function that produces a temporal DAG? We then propose the DAG construction problem on the temporal context, presented below. 
We bring attention to the fact that DAGs are perhaps the most important class of digraphs, given that they not only model many  practical applications (e.g. scheduling~\cite{Sapatnekar.book}, version history~\cite{Bartlang.book}, causal networks like Bayesian Networks~\cite{shmulevich.book}, etc), but also have interesting structural properties that lead to efficient algorithms (e.g. single source shortest paths and longest path~\cite{Cormen.book}, $k$ disjoint paths for fixed $k$~\cite{FHW.80}), and inspires a series of width measures that try to mimic the successful treewidth concept on undirected graphs~(see e.g.~\cite{Getal.14,JRST.82}).
Hence, concerning constructing DAGs, the problem we deal with is the one of assigning a time function to the arcs of a digraph, i.e. give a \emph{temporization} to its arcs, in order to avoid temporal cycles to appear. 
%This can be for safety or control reasons, depending on the context of the network.
%Another interesting problem is to assign a time function to the arcs of a digraph, i.e. give a \emph{temporization} to its arcs, in order to avoid temporal cycles to appear. 
%This can be for safety or control reasons, depending on the context of the network. \Jason{@Andrea can we confidently (or informally) state there's some links to fraud detection as well?}\andrea{I do not know, we need application indeed.} 
% We think that it can be viewed as the temporal, directed version of the well-known task of finding acyclic orientations in static undirected graphs. In that classic problem, directions are assigned to edges so that no cycles are formed, which can be achieved by orienting the edges according to a linear ordering of the nodes. In our case, however, the edge directions are given, and the challenge lies in assigning labels that ensure no cycles arise.
In the following definition, we recall that Type \textsc{T} can be \simplecycle, \weakcycle, or \strongcycle.\\
\begin{minipage}{\textwidth}
\begin{mdframed}
    \underline{T \AcyclicTemporization}\smallskip\\
\textbf{Input:} (Static) digraph $D$ \smallskip\\
\textbf{Question:} Does there exist a temporization $\lambda : A(D) \to 2^\mathbb{N} \setminus \{\emptyset\}$ such that $\tG = (D, \lambda)$ admits no temporal cycles of Type T? 
\end{mdframed} 
\end{minipage}

We do not allow the empty set to be assigned to arcs for the clear reason that doing so for all arcs would be a trivial solution. Note however that we can assume any solution to have exactly one time per arc, since adding more could only create more temporal cycles. This is why, whenever we talk about acyclic temporization, we write simply the number $i$ instead of $\{i\}$ when assigning such a set as the time function of some arc.

\textbf{Our contributions.}
Our results for the non-strict model are summarized in \Cref{fig:results}. Starting with \CycleDetection in \Cref{sec:cycledetection}, we present polynomial-time algorithms to detect \weakcycle{s} and \simplecycle{s}, both using a temporal search algorithm as a subroutine. For \textsc{Strong} \CycleDetection, we prove the problem to be \NP-complete through a reduction from \ThreeSAT. We also provide a complex search algorithm running in \FPT{} time w.r.t. the lifetime parameter, in which search paths are encoded as time values corresponding to the search, 
% instead of the actual vertices and arcs of the path, 
which together with a blocking technique when backtracking allows us to efficiently solve the problem.
% These results also extend to the strict-model.

%Concerning \textsc{Strong} \AcyclicTemporization, we provide a temporization based on arbitrarily ordering the vertices, which avoids \strongcycle{s} in any given digraph, and also avoids \simplecycle{s} in all but the impossible case of graphs with girth $\leq 2$. This temporization also works for avoiding \weakcycle{s} except in the very particular case of girth 4 graphs (which we leave open). \\
%Lastly, when restricting the lifetime of a temporization to be 2, we prove that the previous temporization can be easily adapted to avoid \strongcycle{s}, but surprisingly this does not hold for avoiding \weakcycle{s} and \simplecycle{s}. In fact, we prove both latter problems to be NP-complete through two similar reductions from \textsc{Not All Equal \ThreeSAT}. 
Concerning \AcyclicTemporization, we can always trivially answer \texttt{yes} for \strongcycles by picking any ordering of the vertices, then %labeling 
assigning times to arcs going from smaller to bigger vertices with $1$, and arcs going from bigger to smaller vertices with $2$. This was first noted in~\cite{Bang-JensenBGP22} while dealing with DAG decomposition of static graphs. As for \simplecycle{s} and \weakcycle{s}, if we are allowed to use higher lifetime, we can also construct acyclic temporizations by using an ordering of the vertices. This can always be done for \simplecycle{s}, except when the girth\footnote{The girth of a graph is equal to the minimum length of a cycle.} of $D$ is~2, in which case the answer is trivially \texttt{no}. Similarly, the answer is always \texttt{yes} for \weakcycle{s} when the girth is at least~5, trivially always \texttt{no} when the girth is at most~3, and we leave open the case of girth~4. 
The latter temporization makes a bijection from $A(D)$ to $[m]$, where $m=\lvert A(D)\rvert$. 
If instead the lifetime is bounded,  we prove that \textsc{Simple} \AcyclicTemporization and \textsc{Weak} \AcyclicTemporization become $\NP$-hard for lifetime~2. We do this through reductions from \textsc{Not All Equal \ThreeSAT}. 
We note that these results apply to the non-strict model, as in the case of the strict one, if there are no digons, it is sufficient to give time 1 to all the arcs, that is, the answer is always \texttt{yes}. If there are digons, the answer for \weakcycles is trivially \texttt{no}, while for the other types it is still \texttt{yes} applying the same strategy.

\begin{table}[t]
\centering
		\begin{tabular}{|c || c | c | c |}
        \hline
        \multirow{3}{*}{\textsc{Cycle Def}} & \multicolumn{3}{c|}{\textsc{Problems}}\\
        \cline{2-4}
        & \multirow{2}{*}{\CycleDetection} & \multicolumn{2}{c|}{\AcyclicTemporization}\\
        \cline{3-4}
        &  & Lifetime 2 & Lifetime Unbounded\\
        \hline
        \hline
        \weakcycle & Poly (\Cref{prop:weak_cycle_detection_poly}) & \makecell{\NP-complete \\ (\Cref{thm:weakacy})} & \makecell{\texttt{no} if girth $\leq 3$ \\ \texttt{yes} if girth $\geq 5$ \\ (\Cref{thm:unboweakacy})} \\
        \hline
        \simplecycle & Poly (\Cref{prop:simple_cycle_detection_poly}) & \makecell{\NP-complete \\ (\Cref{thm:simpleacy})} & \makecell{\texttt{no} if girth $\leq 2$ \\ \texttt{yes} if girth $\geq 3$ \\ (\Cref{thm:unbosimpleacy})} \\
        \hline
        \strongcycle & \makecell{\NP-complete \\ (\Cref{theorem:strong_cycle_detection_NPcomplete}) \\ \FPT\ wrt lifetime \\ (\Cref{theorem:strong_cycle_detection_fpt})} & \multicolumn{2}{c|}{\makecell{always \texttt{yes} \\ (\Cref{prop:strongtemp})}}\\
        \hline
        \end{tabular}
\caption{Main results for our problems on \CycleDetection and \AcyclicTemporization, concerning \weakcycle{s}, \simplecycles, and \strongcycles.
        % Concerning the results depending on the girth of the graph, note that the girth of a given graph can be computed in polynomial time and thus the associated problems are in P. 
			\label{fig:results}}
\end{table}

\textbf{Related Works.} We are not aware of a systematic study of cycles in temporal graphs. We can find in the literature studies about simple-cycles, for instance concerning Eulerian temporal cycles~\cite{BumpusM23,MarinoS23}, and Hamiltonian cycles, also referred to as temporal vertex exploration returning to the base~\cite{AkridaMSR21} (where the latter is a constrained version of the temporal vertex exploration problem where there is no need to go back to the starting vertex~\cite{ErlebachS23,Erlebach0K21}). On the other hand, as far as we know, surprisingly, we are the first ones to introduce the notion of \weakcycle and \strongcycle.

Detecting a cycle in static graphs can be easily done by applying a \textsc{Breadth-First Search} (BFS) or a \textsc{Depth-First Search} (DFS) from any vertex. Indeed, when the search explores an edge which leads to an already visited vertex, then a cycle has been detected, and when this does not occur, then no cycle exists. In digraphs, a similar idea works, although instead of an already visited vertex triggering detection, the vertex has to be in the current search path as well. 
In \cite{xuan2003computing}, (polynomial-time) search algorithms are presented for temporal graphs. Among these, one computes earliest arrival paths from the root vertex to the other vertices, or in other words, it computes earliest arrival times (earliest among all possible temporal paths) from the root to the other vertices. Informally, the search progresses by selecting earliest incident edges such that they obey the temporal order of the created temporal paths. In~\cite{wu2014path}, this result is presented again, but complemented by a similar algorithm for computing latest departure times between vertices. 

Concerning \AcyclicTemporization, we highlight that this falls into the so-called network realization problem framework, where we are given a static graph and we have to assign time to the arcs in order to meet some property. Some of the properties considered in the literature are: ensure reachability~\cite{klobas_et_al:LIPIcs.MFCS.2022.62}; and meet exact/upper bounds on the fastest path durations among its vertices on periodic temporal graphs~\cite{klobas_et_al:LIPIcs.SAND.2024.16,mertzios2024realizingtemporaltransportationtrees}. Another close relation to this notion of acyclic temporization is the one of \emph{Good edge-labeling}~\cite{BCP13}. A \emph{labeling} of the edges of a given simple undirected graph $G$ is an assignment of a real number to each edge of $G$. It is said to be \emph{good} if, for any pair of vertices $u,v\in V(G)$, there do not exist two non-decreasing $u,v$-paths, with respect to the edge labels. In particular, labels can be assumed to be distinct, i.e. strict and non-strict cases are equivalent in this context. Note that this notion is similar, but not equivalent, to the case of \textsc{Weak} \textsc{Acyclic Temporization}. 
In~\cite{BCP13}, the authors use the notion of good edge-labeling to prove that there exist particular optical networks %(directed acyclic graphs $D$ such that for each pair of vertices there is at most one directed path linking them) 
and set of requests to be assigned wavelengths  %(directed paths such that each arc of $D$ contains at most two of these directed paths) 
such that, if one wants to assign distinct wavelengths to requests sharing an arc, then the number of wavelengths can be arbitrarily large.

Finally, let us mention the problem of computing a temporal feedback edge set as discussed in~\cite{HaagMNR22}, which also aims to achieve acyclic temporal graphs. However, unlike our approach of assigning suitable times to ensure acyclicity, their method considers a given temporal graph and focuses on removing a subset of temporal edges (referred to as time-edges) or edges (referred to as connection sets) to eliminate all simple cycles.

\textbf{Structure of the paper.} In Section~\ref{sec:preliminaries}, we present our notation and definitions. In Section~\ref{sec:cycledetection}, we present our results about detecting cycles, and in Section~\ref{sec:temporization} the results about acyclic temporization. 
Results marked with a $\normalfont{(\star)}$ indicate that the proof (and/or corresponding lemmas etc.) are moved to the appendix, or only a proof sketch is provided. 

\section{Preliminaries}
\label{sec:preliminaries}

Given a digraph $D$, a \emph{walk} in $D$ is a sequence $W = (v_1,e_1,v_2,\ldots,v_q,e_q,v_{q+1})$ of alternating vertices and arcs of $D$ where $e_i = v_iv_{i+1}$ for each $i\in [q]$. It is a \emph{path} if $v_1,\ldots,v_{q+1}$ are all distinct and a \emph{cycle} if $v_1,\ldots,v_q$ are all distinct and $v_1 = v_{q+1}$. We denote by $V(W)$ the set $\{v_1,\ldots, v_{q+1}\}$ and by $A(W)$ the set $\{e_1,\ldots, e_q\}$. It is said that $W$ has \emph{length} $q$ and \emph{order} $q+1$. In this paper, we work on simple digraphs, so we can omit the arcs from the sequence, writing $W = (v_1,v_2,\ldots,v_q,v_{q+1})$ instead.
Given a temporal directed graph $\mathcal{D} = ( D,\lambda)$, the \emph{set of vertices} of $\mathcal{D}$ is equal to $V(D)$, the \emph{set of arcs} of $\mathcal{D}$ is equal to $A(D)$, the set of \emph{temporal vertices} of $\mathcal{D}$ is equal to $V(D)\times [\tau]$, and the set of \emph{temporal arcs} of $\mathcal{D}$ is equal to $\{(e,t)\mid e\in A(D) \mbox{ and }t\in \lambda(e)\}$. These are denoted, respectively, by $V(\mathcal{D})$, $A(\mathcal{D})$, $V^T(\mathcal{D})$, and $A^T(\mathcal{D})$. Given vertices $v_1,v_{q+1}\in V(D)$, a \emph{temporal $v_1,v_{q+1}$-walk} in $\mathcal{D}$ is defined as a sequence of vertices and times  $W = (v_1, t_1,v_2,\cdots,t_q,v_{q+1})$ such that, for each $i\in [q]$, there exists $e_i = v_iv_{i+1} \in A(D)$, $t_i\in \lambda(e_i)$, and $t_i\le t_{i+1}$. 
An equivalent definition exists concerning temporal edges. It is said to be \emph{strict} if $t_i<t_{i+1}$ for every $i\in [q]$, and non-strict if $t_i = t_{i+1}$ for some $i$. It is called a temporal  \emph{$v_1,v_{q+1}$-path} if all vertices are distinct. 
 We also say that $W$ \emph{starts or departs at time $t_1$} and \emph{finishes or arrives at time $t_q$}. 
The set $\{v_1,\ldots, v_{q+1}\}$ is denoted by $V(W)$ and the set $\{e_1,\ldots,e_q\}$, by $A(W)$. Additionally, the set $\{(e_i,t_i)\mid i\in[q]\}$ is denoted by $E^T(W)$.
%Moved the def of cycles in the intro

We write EAT$(u, v)$ to be the \emph{earliest arrival time} from vertex $u$ to vertex $v$, defined as the earliest arrival time among all temporal paths from $u$ to $v$. Special cases include EAT$(u, u) = 0$, and EAT$(u, v) = + \infty$ if $u$ cannot reach $v$. Similarly, LDT$(u, v)$ is the \emph{latest departure time} from vertex $u$ to vertex $v$, defined as the latest departure time among all temporal paths from $u$ to $v$. Special cases include LDT$(u, u) = \tau$, and LDT$(u, v) = - \infty$ if $u$ cannot reach $v$. 
% In this paper, we also use slight adaptations, such as EAT$(u, v, P, t)$ which denotes the EAT from $u$ to $v$ in the temporal digraph restricted to search path $P$ among all temporal paths departing at time at least $t$. 
As mentioned in the introduction, earliest arrival times and latest departure times can be computed in polynomial time \cite{xuan2003computing,wu2014path},. 
% This remains true for the aforementioned adaptations of EAT and LDT: for example for EAT$(u, v, P, t)$, it suffices to compute EAT$(u, v)$ in the temporal graph induced by search path $P$ and where all times smaller than $t$ are removed. 
We use these algorithms as a black box for \CycleDetection.

\section{Cycle detection}
\label{sec:cycledetection}

In this section we describe our results for the \CycleDetection problem. By computing earliest arrival times, we obtain the first two polynomial-time results for \simplecycles and \weakcycles. In the remainder, namely Section~\ref{sec:detect-strong}, we prove hardness for \textsc{Strong} \CycleDetection and give an \FPT algorithm wrt the lifetime $\tau$.

% The first two are easy.
% For the first one, for each vertex x, we create a twin y and we check whether x can reach y using earliest arrival path computation. If yes then choose as C the smallest cycle involving $x$ in the path formed by the used edges.
% For the second one, similar. For each pair $x$, $y$ check whether $x$ can reach $y$ and viceversa. If the discovered path share a vertex $z$, then use $z$ as the new $y$. 

% The third (=fourth) we do not know.

% \begin{proposition}
% \label{prop:finding_onecycle_poly}
% Let $\mathcal{G} = (G,\lambda)$ be a temporal directed graph. One can decide, in polynomial time, whether $\mathcal{D}$ contains a {\simplecycle} or a {\weakcycle}. This holds also if only strict temporal walks can be used. 
% \end{proposition}

\begin{proposition}
\label{prop:weak_cycle_detection_poly}
$\normalfont{(\star)}$ \textsc{Weak} \CycleDetection is polynomial-time solvable.
\end{proposition}

\begin{proposition}
\label{prop:simple_cycle_detection_poly}
$\normalfont{(\star)}$ \textsc{Simple} \CycleDetection is polynomial-time solvable.
\end{proposition}

% Concerning strong-cycles, we observe that in temporal graphs of temporality 1 (i.e. one label per edge, or so-called \textit{simple} temporal graphs), one can easily detect strong-cycles since no strict or proper strong-cycles exist, and the only non-strict cycles that can exist have the same label on each edge meaning the problem reduces down to a cycle detection in each snapshot. A natural question then would be to consider the next step, being temporal graphs of temporality 2, i.e. admitting at most 2 labels per edge.

\subsection{Detecting strong-cycles}
\label{sec:detect-strong}

The algorithms detecting \simplecycles and \weakcycles can efficiently use the black box for EAT because, intuitively, the temporal paths corresponding to these EAT concatenate nicely into a cycle structure when there is only one or two vertices that need to reach themselves or each other. This nice concatenation cannot be ensured when multiple such vertices and thus multiple temporal paths exist, which is the case for \strongcycles. 

This difficulty makes \textsc{Strong} \CycleDetection \NP-complete, as shown next by reducing from \textsc{3-SAT}. Interestingly, the lifetime of the digraph resulting from the reduction depends on the size of the formula. This is not by chance, as indeed, in the remainder, we prove that the problem is FPT with respect to the lifetime.
%Nevertheless, we note that a core idea behind the FPT algorithm for detecting \strongcycles is very similar to the algorithm for detecting \simplecycles, in that it starts from an arc $(v, r)$ and tries to reach $v$ from $r$ before the latest time on that arc.\\
%Before presenting the FPT algorithm however, let us first present why FPT is the best we could hope for in terms of tractability.

\subsubsection{Hardness of detecting \strongcycle}
In order to prove hardness of detecting \strongcycle, we present first the \emph{auxiliary cycle} structure that will be useful when assigning times to the arcs of the constructed digraph. 
% We then prove that this is a \strongcycle. 
Then, we propose our reduction.

%In this section, we prove that \CycleDetection concerning {\strongcycle}s is \NP-complete by reducing from \textsc{3-SAT}. For this, we present first the \emph{auxiliary cycle} structure that will be useful when assigning times to the arcs of the constructed digraph.

% \begin{definition}
%     The \emph{auxiliary cycle of order $n$} 
%     % \davi{was it introduced before? \Jason{I thought in some paper by Mertzios et al. but I checked and it's something else, so I'd guess no.}} 
%     is the directed graph whose vertex set is $\{x, v_1, v_2, \dots, v_{n-1}\}$ and arc set is $\{e_1 = xv_1\} \cup \{e_{i+1} = v_iv_{i+1} \mid 1 \leq i \leq n-2\} \cup \{e_n = v_{n-1}x\}$, with $\lambda(e_i) = \{n-i+1, 2n-i+1, 3n-i+1, \dots, (n-1)n-i+1\}$ for each $2 \leq i \leq n$ and $\lambda(e_1) = \{0, n, 2n, 3n, \dots, n(n-1)\}$. \davi{is it necessary for $x$ to reach itself?}
% \begin{figure}[h!]
%     \centering
%     \input{figures/Section 2/1 - auxiliary_cycle}
%     \caption{Auxiliary cycle of order 5. \davi{a bigger example?}}
% \end{figure}
% \end{definition}

\begin{definition}
    The \emph{auxiliary cycle} of order $n$ 
    % \davi{was it introduced before? \Jason{I thought in some paper by Mertzios et al. but I checked and it's something else, so I'd guess no.}} 
    is the temporal digraph whose vertex set is $\{v_0, v_1, v_2, \dots, v_{n-1}\}$ and arc set is $\{e_0 = %(v_{n-1},v_0)\} changed arc
    v_{n-1}v_0\} \cup \{e_{i} = %(v_{i-1},v_i) changed arc 
    v_{i-1}v_i\mid 1 \leq i \leq n-1\}$, with $\lambda(e_0) = \{0, n, 2n, 3n, \dots, (n-1)n\}$ and $\lambda(e_i) = \{n-i, 2n-i, 3n-i, \dots, (n-1)n-i\}$ for each $1 \leq i \leq n-1$. See \Cref{fig:aux_cycle} for an example.
\end{definition}

\begin{figure}[h!]
    \centering
    \includegraphics[scale=0.9]{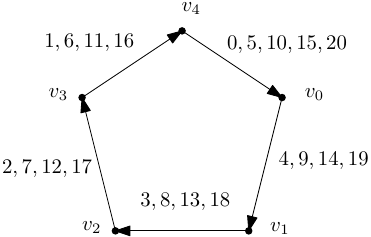}
    \caption{Auxiliary cycle of order 5.}
    \label{fig:aux_cycle}
\end{figure}

We now prove that every auxiliary cycle is a \strongcycle. To this end, we have:

\begin{proposition}
\label{prop:aux_cycle_unicity_and_disjointness}
    $\normalfont{(\star)}$ Given an auxiliary cycle $\mathcal{A}$ of order $n$, there exists exactly one \tpath{$v,v$} for each $v \in V(\mathcal{A})$. In particular, given $v_i$ and $v_k$ such that $i \neq k$ and %$i,j \neq n-1$
    $n-1 \notin \{i,j\}$, these paths do not share temporal arcs. %times\julio{??? temporal arcs?}.
\end{proposition}

Let us refer to the times that each vertex $v_i \neq v_{n-1}$ requires to reach itself in the auxiliary cycle, as $L^\circlearrowright(v_i) = \{(n-1) -i, 2(n-1)-i, 3(n-1)-i, ..., n(n-1)-i\}$.

Note that vertex $v_{n-1}$ admits exactly one temporal path as well, and that it uses times $0 \cup L^\circlearrowright(v_0) \setminus \max(L^\circlearrowright(v_0))$. Together with \Cref{prop:aux_cycle_unicity_and_disjointness}, we thus have that any auxiliary size is a \strongcycle.

% \davi{labels in the \tpath{$v_i, v_i$} are $(n-i), (n-i) + (n-1), \dots, (n-i) + (n-1)^2$.}
% \davi{TODO: exactly one path and the paths are disjoint for the vertices $v_i$ ($x$ uses $v_1$'s labels!!}

\begin{theorem}
    \label{theorem:strong_cycle_detection_NPcomplete}
	$\normalfont{(\star)}$ \textsc{Strong} \CycleDetection is \NP-complete.
\end{theorem}

\begin{proof}\emph{(Sketch)}
    \textsc{Strong} \CycleDetection is in \NP, because a solution subgraph $\mathcal{C}$ can be verified to be a cycle in the underlying graph, and deciding whether each vertex reaches itself can be done by checking whether EAT$(v, u)$ in $\mathcal{C}$ is at most $\max(\lambda(u,v))$, for each arc $%(u,v) changed arc
    uv \in A(\mathcal{C})$, similarly to \Cref{prop:weak_cycle_detection_poly}.
	
	To prove this problem is \NP-hard, we reduce \ThreeSAT to it. Let the generic instance of \ThreeSAT be the CNF formula $\phi$ of $n$ variables $x_0, x_1, ..., x_{n-1}$ and $m$ clauses $C_0, C_1, ..., C_{m-1}$. Let the literals of clause $C_i$ be denoted as $\ell_{i,1}, \ell_{i,2}$, and $\ell_{i,3}$. 
	Let us build an instance of \CycleDetection as the temporal digraph $\mathcal{D}(\phi)$ as follows. Initially, add three auxiliary cycles $\mathcal{A}^1$, $\mathcal{A}^2$, and $\mathcal{A}^3$, all of order $4m+1$. Let the corresponding vertices of $\mathcal{A}^1$, $\mathcal{A}^2$ and $\mathcal{A}^3$ be referred to as $v^1_i$, $v^2_i$, and $v^3_i$ respectively, for every $i\in\{0,\ldots, 4m\}$. Note that for any three vertices $v^1_i$, $v^2_i$, and $v^3_i$,  $L^\circlearrowright(v^1_i) = L^\circlearrowright(v^2_i) = L^\circlearrowright(v^3_i)$ Let us simply refer to these times as $L^\circlearrowright(v_i)$ instead.
	Now, for each $i \in \mathbb{N}$, merge\footnote{We define merging of vertices in temporal digraphs as in static digraphs, and times on arcs of pre-merged vertices remain on corresponding arcs of post-merged vertices.} the three vertices $v^1_{4i}$, $v^2_{4i}$, and $v^3_{4i}$, and refer to this merged vertex as $v_{4i}$ (see \Cref{fig:reduction_strongcycledetection}). Note that this also merges vertices $v^1_{4m}$, $v^2_{4m}$, and $v^3_{4m}$ into vertex $v_{4m}$.
	For each clause $C_i$, add time 0 to arcs %$(v^1_{4i+2}, v^1_{4i+3})$, $(v^1_{4i+3}, v_{4(i+1)})$, $(v^2_{4i+1}, v^2_{4i+2})$, $(v^2_{4i+3}, v_{4(i+1)})$, $(v^3_{4i+1}, v^3_{4i+2})$, and $(v^3_{4i+2}, v^3_{4i+3})$. changed arc
    $v^1_{4i+2}v^1_{4i+3}$, $v^1_{4i+3}v_{4(i+1)}$, $v^2_{4i+1}v^2_{4i+2}$, $v^2_{4i+3}v_{4(i+1)}$, $v^3_{4i+1}v^3_{4i+2}$, and $v^3_{4i+2}v^3_{4i+3}$.
	Finally, for each literal $\ell_{i, j}$ corresponding to some variable $x_k$, and each literal $\ell_{g, h}$ which corresponds to $\neg x_k$, remove from arc %$(v_{4g}, v^h_{4g+1})$ changed arc
    $v_{4g}v^h_{4g+1}$ all times from $L^{\circlearrowright}(v_{4i+j})$. This concludes the transformation. 
	
	\begin{figure}
		\begin{subfigure}[b]{0.49\textwidth}
        \centering
			\includegraphics[scale=.8]{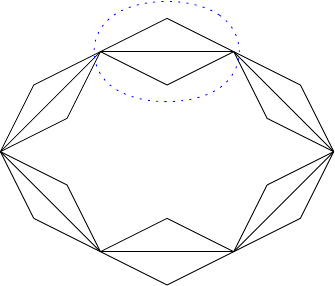}
			\caption{Transformation after merging vertices.
				\label{fig:reduction_strongcycledetection_1}}
		\end{subfigure}%
		\hfill
		\begin{subfigure}[b]{0.49\textwidth}
			\includegraphics{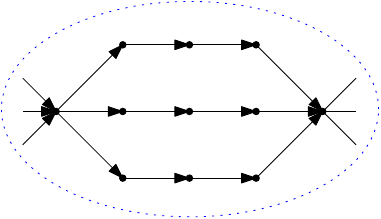}
			\caption{Paths between vertices $v_{4i}$ and $v_{4(i+1)}$.
				\label{fig:reduction_strongcycledetection_2}}
		\end{subfigure}%
		\caption{Example of the transformation for a \ThreeSAT formula of $m=6$ clauses. Each clause is represented as three paths between merged vertices in the created \CycleDetection instance (see the blue dotted selection in \Cref{fig:reduction_strongcycledetection_1}, and a more detailed view of these paths in \Cref{fig:reduction_strongcycledetection_2}).  
		For clarity, times on the arcs have been omitted.
		\label{fig:reduction_strongcycledetection}}
	\end{figure}
	
	In the appendix we prove that a positive instance for \ThreeSAT remains a positive instance for \CycleDetection after the transformation (\ThreeSAT $\implies$ \CycleDetection), and vice versa, that a positive instance of \CycleDetection in the transformed instance implies a positive instance of \ThreeSAT before the transformation (\CycleDetection $\implies$ \ThreeSAT).  The key idea for both directions is that a literal $\ell_{i, j}$ is true, if and only if path $(v^j_{4i+1}, v^j_{4i+2}, v^j_{4i+3})$ is part of a \strongcycle.
 \end{proof}

\subsubsection{Fixed-parameter tractability wrt lifetime}

To detect strong-cycles, a modified depth-first search is employed on every outgoing arc %$a_r = (r,v)$ changed arc
$a_r = rv$ of every possible root vertex $r$ (see \Cref{algo:full_cycle_detection} in the appendix). This search aims to iteratively construct a strong-cycle by exploring arcs from $r$, creating and extending a \textit{search path} $P$, until reaching $r$ again (see \Cref{algo:search_full_cycle} in the appendix). Along the search, time values corresponding to temporal paths along the search path are updated (see \Cref{algo:extend} in the appendix). Throughout the search, unsuccessful search paths backtrack and employ a ``blocking'' mechanism of said time values on the backtracked arcs, which effectively restrains running time to be superpolynomial in the lifetime parameter $\tau$ only. 

Let us first present the main data structures used to keep track of temporal paths: let \textit{root timetable} $T_r$ and \textit{path timetable} $T_P$ be defined as two arrays of size $\tau+1$, containing time values $\in \{0\} \cup [\tau]$, all initialized to $0$. Another important data structure used is the \textit{blocking matrix} $B$, which we will present later.
The values of the timetables change over the course of the search, and revert back to previous values when the search backtracks. (In the pseudo code, this is done through recursion and copying of the timetables.)

%Each value $B[x][y]$ of the blocking matrix, being \texttt{true} or \texttt{false}, then respectively encodes whether arc $a$ with $\texttt{order}(a) = x$, blocks search paths with timetables $T_r$ and $T_P$ where $\texttt{order}(T_r, T_P) = y$, or if $a$ allows the search to extend through it. 

The timetables will represent the following throughout the search: 

\begin{itemize}
	\item Root timetable $T_r$: any non-zero value $T_r[x] = y$ represents ``The earliest arrival time from root vertex $r$, starting with a time $t \geq x$, following along search path $P$, and ending at the last vertex of the search path $v$, is $y$.''\footnote{This is denoted as $EAT_P^{\geq x}(r,v)$ in the pseudo code comments.} A zero value indicates no temporal path from $r$ to $v$, along $P$, and starting with a time at least $x$, exists;
	\item Path timetable $T_P$: any non-zero value $T_P[x] = y$ represents ``The earliest arrival time from some vertex $u \in P$, following along search path $P$, and ending at the last vertex of the search path $v$, is $y$; and the latest departure time from root vertex $r$, following along $P$ and ending at $u$, is $x+1$.''\footnote{These are denoted as $EAT_P(u,v)$ and $LDT_P(r,u)$ in the pseudo code comments.} In other words, the $x$ value represents a ``deadline'' by which $u$ needs to reach $r$, and thus if at some point of the search $T_P[x] > x$, that means the corresponding search path $P$ cannot result in a strong-cycle, due to some vertex $u \in P$ not being able to reach itself via $P$ and $r$. A zero value $T_P[x] = 0$ indicates no vertex in $P$ admits a latest departure time from $r$, along $P$, of $x+1$; and $T_P[\tau]$ is for the special case of $u = r$, as the latest departure time from $r$ to $r$ can be considered later than $\tau$, or $+\infty$.
\end{itemize}

Each search starts by setting for all $t$, $T_r[t] = \min \ell \in \lambda(a_r) : \ell \geq t$, i.e. the smallest time on arc $a_r$ such that it is greater or equal to $t$, or keep $T_r[t] = 0$ if no such time exists. 
%These values get updated throughout the search, so that at any point, it represents the earliest arrival time starting at time $t$ at the root $r$, following along the search path $P$, and ending at the last vertex of the search path, which we will denote as $v$. 
%Informally, the root timetable represents all possible temporal paths from $r$ to $v$ through $P$. Keeping track of it is crucial for obtaining our technical ``search path independence'' lemma (\Cref{lemma:search_path_independence}).

Then, whenever an arc %$a =(u,v)$ changed arc 
$a =uv$ is explored to extend the search, we update $T_P$ by, for each non-zero value $T_P[x] = y$, replacing it by the smallest time $y' > y$ among $\lambda(a)$, and we set $T_P[\max t : T_r[t+1] \neq 0] = \min(\lambda(a))$. (If this case already has a non-zero value, we keep the largest of the two values.) This represents the reachability from/to vertex $u$ in search path $P$. Root timetable $T_r$ is also updated by, for each non-zero value $T_r[x] = y$, replacing it by the smallest time $y' > y$ among $\lambda(a)$ (or 0 if no such time exists).
%, encoding that if $u$ needs to reach itself in the strong-cycle that is being constructed through search path $P$, it needs to do so by reaching back to $r$ before time $\max t : T_r[t] \neq 0$, because that is the latest departure time from $r$ to $u$ through $P$. 
%The time value stored represents the earliest arrival time from $u$ to $v$, i.e. to the last vertex of the search path. These values get updated as the search progresses to always represent this earliest arrival time.
%A long as for all $t$, we have $T_P[t] \leq t$, it remains possible (concerning the self-reachability of vertices in $P$) to find a strong-cycle through the current search path $P$.
% In the case of $T[t]$ already containing some value $a$ while we try to add another value $b$, then we keep the maximum of the two values.

In the best case scenario, arcs are explored until the search finds root vertex $r$ again (or another vertex in the search path $P$), and $T_P$ can be updated correctly, i.e. for all $t$, $T_P[t] \leq t$. In this case, it signifies that vertices can reach $r$ in time to then reach back to themselves by departing from $r$. Thus, all vertices can reach themselves through the underlying structure of the successful search path. This underlying structure is ensured by design to be a cycle since a search path is extended until it loops back to root vertex $r$ (or some other vertex of the search path).

However, when the search encounters an issue, we use the \textit{blocking matrix} $B$, which is an $m \times 4^{(\tau+1)^2}$ matrix, containing boolean values, initially all set to \texttt{false}.
In the blocking matrix, the only modifications allowed are changing \texttt{false} values to \texttt{true}, which happens when the search unsuccessfully tries to extend, and when the search backtracks. Regarding the size of this matrix, we assume some order exists on $A$, the arcs of the temporal graph, which corresponds to the first dimension of size $m = |A|$; and we assume some order exists on $\mathcal{T}$, the collection of all possible states of timetables $(T_r, T_P)$, corresponding to the other dimension of size $4^{\tau^2} = |\mathcal{T}|$. 
To access and modify $B$, we thus use functions $\texttt{order} : A \rightarrow [m]$ and $\texttt{order} : \mathcal{T} \rightarrow [4^{\tau^2}]$. 

The value $B[\texttt{order}(a)][\texttt{order}(T'_r, T'_P)]$ being \texttt{true} indicates that searches with corresponding timetables $T_r$ and $T_P$ s.t. $T_r = T'_r$ and $T_P = T'_P$, are blocked on $a$, i.e. such searches cannot extend through $a$, due to some previous search with these timetables already having explored through $a$ unsuccessfully. A \texttt{false} value instead allows the search to extend through $a$.

When some non-zero value $T_P[x] = y$ in the path timetable fails to update correctly while extending through some arc $a = uv$, i.e. trying to replace it by the smallest time $y' > y$ among $\lambda(a)$ results in a value $T_P[x] > x$ (or no such time $y'$ exists) then the search stops the exploration through this arc and continues through another arc $uv'$. It then sets $B[\texttt{order}(a)][\texttt{order}(T_r, T_P)] = \texttt{true}$, where $T_r$ and $T_P$ are the timetables before the failed update. When the search cannot extend through any other arc $uv'$, the search backtracks to the previous vertex $w$ and arc %$(w,u)$ changed arc 
$wu$ now blocks the last timetables $T_r$ and $T_P$ that extended through the arc (i.e. the timetables as they were before exploring $wu$). Thus, the blocking matrix updates $B[\texttt{order}(wu)][\texttt{order}(T_r, T_P)] = \texttt{true}$. 
%Blocked timetables prevent search paths extending through arcs if the corresponding timetables are identical to a previous search path that has explored this arc already. 

This blocking mechanism allows for only a limited amount of explorations per arc, since every exploration of an arc $a$ which doesn't end in a strong-cycle, will backtrack and change a \texttt{false} in the blocking matrix row $B[a]$ to \texttt{true}. Since $B[a]$ contains $4^{(\tau+1)^2}$ values, this can only occur a limited amount of times before all future searches are blocked to go through $a$, and the result follows.

We present in the appendix a corresponding implementation in pseudo code, a figure of key operations of the search algorithm on (part of) an example temporal graph, and the formal proofs of correctness and complexity, of which an important part is the following technical lemma which proves that the vertices and arcs of a path are irrelevant, and that only the timetables matter.
% Concerning implementation specifics, we choose a list structure for the search path $P$.

% \begin{lemma}[Search path independence]
% 	\label{lemma:search_path_independence}
% 	If a search path $P$ and another search path $Q \neq P$ both start at the same root $r$, arrive at the same vertex $u$, and have the same timetables $(T_r, T_P)$, then on any arc %$a = (u, v)$ changed arc
% 	$a = uv$ for some $v \in V \setminus (P \triangle Q)$, both search paths extend in the exact same manner, i.e. the timetables remain identical after updating. 
% \end{lemma}

\begin{restatable}[]{lemma}{searchpathindependence}
\label{lemma:search_path_independence}
$(\star)$ If a search path $P$ and another search path $Q \neq P$ both start at the same root $r$, arrive at the same vertex $u$, and have the same timetables $(T_r, T_P)$, then on any arc %$a = (u, v)$ changed arc
	$a = uv$ for some $v \in V \setminus (P \triangle Q)$, both search paths extend in the exact same manner, i.e. the timetables remain identical after updating. 
\end{restatable}

Altogether, this results in the following.

\begin{restatable}[]{theorem}{strongcycledetectionfpt}
	\label{theorem:strong_cycle_detection_fpt}
$(\star)$ \textsc{Strong} \CycleDetection is fixed-parameter tractable with the parameter being the lifetime.
\end{restatable}

%%%%%%%%%%%%%%%%%%%%%%%%%%%%%%%%%%%%%%%%%%%%%%%%%%%%%%%%%%%%%%%%%%%%%%%%%%%%%%%%%%%%%%%%%%%%%%%%%%%%%%%%%%%%%%%%%%%%%%%%%%%%%%%%%%%%%%%%%%%%%%%
%%%%%%%%%%%%%%%%%%%%%%%%%%%%%%%%%%%%%%%%%%%%%%%%%%%%%%%%%%%%%%%%%%%%%%%%%%%%%%%%%%%%%%%%%%%%%%%%%%%%%%%%%%%%%%%%%%%%%%%%%%%%%%%%%%%%%%%%%%%%%%%
%%%%%%%%%%%%%%%%%%%%%%%%%%%%%%%%%%%%%%%%%%%%%%%%%%%%%%%%%%%%%%%%%%%%%%%%%%%%%%%%%%%%%%%%%%%%%%%%%%%%%%%%%%%%%%%%%%%%%%%%%%%%%%%%%%%%%%%%%%%%%%%

\section{Acyclic Temporization}
\label{sec:temporization}

Given a directed graph $D$, a \emph{temporization of $D$} is an assignment of a non-empty time function $\lambda$ to each arc of $D$. In this section, for each type of temporal cycle, we are interested in finding temporizations that do not contain such cycles. 
% First, observe that if only strict walks are allowed, then this problem becomes trivial for all types of cycles as we can simply assign only the timestep 1 to all arcs. Therefore, we do not need to treat this case anymore in the next sections. Additionally,
In the following, we first give an easy solution for \textsc{Strong} \AcyclicTemporization, we then focus in Section~\ref{sub:simple} and in Section~\ref{sec:weak-cycles} respectively on \simplecycles and \weakcycles. Both these sections first analyses the unbounded lifetime case, i.e. we can use as many time values as we want to solve \AcyclicTemporization, and then, we focus on the bounded lifetime case, which turns out to be \NP-complete for both for $\tau=2$ using similar reductions.

In this section, we often use the following temporization, referred to as \textit{\lexicographicTemporization}.

\begin{definition}[\lexicographicTemporization]
    Assign an arbitrary order to $V(D)$. For all arcs $(u, v)$ such that $u>v$ (suppose there are $m'$) assign in an incremental manner one time per arc in lexicographic order, thus assigning times $1$ to $m'$ to these arcs. For the other arcs, being arcs %$(u,v)$ changed arc 
    $uv$ such that $u<v$, assign in an incremental manner one time per arc in reverse lexicographic order, starting from time $m'$ (and thus ending with time $m$).
\end{definition}

 \begin{figure}[h]
\begin{center}
%\input{figs/variable_gadgets}
% \begin{tikzpicture}[scale=1]
%   \pgfsetlinewidth{1pt}
%   \pgfdeclarelayer{bg}    
%    \pgfsetlayers{bg,main} 
  
%   \tikzset{vertex/.style={circle, minimum size=0.15cm, draw, inner sep=1pt, fill=black}}

%   \node[vertex,label=90:$v_1$] (v1) at (0,0) {};
%   \node[vertex,label=90:$v_2$] (v2) at (2,0) {};
%   \node[vertex,label=90:$v_3$] (v3) at (4,0) {};
%   \node[vertex,label=90:$v_4$] (v4) at (6,0) {};
%   \node[vertex,label=90:$v_5$] (v5) at (8,0) {};
  
%   \begin{pgfonlayer}{bg}    % select the background layer
%       %\draw (a1) edge [bend left] node [midway, fill=white] {$4$} (xi);
%       \draw[->, >=stealth] (v1) to["8"](v2); \draw[->, >=stealth] (v1) to ["7", bend left] (v3); \draw[->, >=stealth] (v2) to ["6", bend left] (v5); \draw[->, >=stealth] (v3) to["5"] (v4);
      
%       \draw[<-, >=stealth] (v2) to ["1", bend right] (v3);\draw[<-, >=stealth] (v1) to ["2", bend right] (v4);\draw[<-, >=stealth] (v2) to ["3", bend right] (v5); \draw[<-, >=stealth] (v4) to ["4", bend right] (v5);

%       %[->, >=stealth] 
%   \end{pgfonlayer}
%   \end{tikzpicture}
\includegraphics[scale=.9]{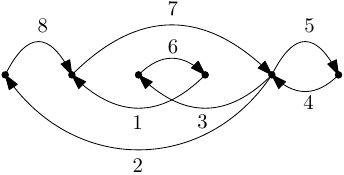}
\caption{Lexicographic temporization with $m=8$ and $m' = 4$.
}\label{fig:lexicographic_temporization}
\end{center}
\end{figure}

The following property for lexicographic temporizations holds.

\begin{lemma}    \label{lemma:lexicographicTemporizationJourneysAtMostLengthTwo}
    Any digraph with the \lexicographicTemporization results in a temporal graph which has temporal paths of length at most two. 
\end{lemma}

\begin{proof}
    Consider two arcs %$(u,v)$ and $(v,w)$ changed arc
    $uv$ and $vw$
    such that $u < v < w$ in the arbitrary ordering. The \lexicographicTemporization ensures that $\lambda(vw) < \lambda(uv)$, meaning no temporal path exists using these successive arcs. The same holds for arcs %$(w, v)$ and $(v, u)$ changed arc
    $wv$ and $vu$, and for arcs %$(u, w)$ and $(w, v)$ changed arc
    $uw$ and $wv$, again with $u<v<w$. The remaining case is arcs %$(w,u)$ and $(u,v)$ changed arc 
    $wu$ and $uv$ which the temporization does allow to form a temporal path, but the temporal path is then stuck by the previous case analysis. 
 \end{proof}

%The following easy proposition allows us to concentrate only on {\simplecycle}s and {\weakcycle}s in the remainder of the section, as for \strongcycle we can use the \lexicographicTemporization or the same strategy in~\cite{Bang-JensenBGP22} to obtain a temporization which will use only two values, i.e. providing a \texttt{yes} answer for \AcyclicTemporization already for $\tau=2$.

%\begin{proposition}
%Let $D$ be a directed graph. Then there always exists a temporization $\lambda:E(D)\rightarrow 2^{\mathbb{N}}$ of $D$ such that $(D,\lambda)$ contains no {\strongcycle}s.
%\end{proposition}
%\begin{proof}
%Apply the \lexicographicTemporization on $D$. By construction, there is no cycle using only $(u,v)$ arcs such that $u<v$, nor using only $(u,v)$ arcs such that $u>v$. Now, if $C$ is a cycle in $D$, by \Cref{lemma:lexicographicTemporizationJourneysAtMostLengthTwo}, $C$ cannot contain more than two vertices. Suppose one of the vertices, say $v_1$ is able to reach itself through temporal path with times $\ell_1, \ell_2$, with $\ell_1 < \ell_2$. Then, since the arcs are assigned only one time, clearly $C$ does not contain any non-trivial \tpath{$v_2,v_2$}.
% \end{proof}

%We obtain the same result if $(u, v)$ arcs such that $u<v$ are assigned time 1, and time 2 otherwise.

The following easy proposition allows us to concentrate only on {\simplecycle}s and {\weakcycle}s in the remainder of the section, as for \strongcycles we can use the same strategy as in~\cite{Bang-JensenBGP22} to obtain a temporization which will use only two values, i.e. providing a \texttt{yes} answer for \textsc{Strong} \AcyclicTemporization already for $\tau=2$. In particular, it is enough to order the vertices of the input digraphs and give times to the arcs as follows: $uv$ arcs such that $u<v$ are assigned time 1, and time 2 otherwise.

\begin{proposition}
$\normalfont{(\star)}$ Let $D$ be a directed graph. Then there always exists a temporization $\lambda:E(D)\rightarrow 2^{[2]}$ of $D$ such that $(D,\lambda)$ contains no {\strongcycle}s.
\label{prop:strongtemp}
\end{proposition}

\subsection{Simple-cycles}
\label{sub:simple}

%\subsubsection{Unbounded lifetime}
The following result characterizes the answer to \textsc{Simple} \AcyclicTemporization wrt the girth of the input digraph.

\begin{lemma}
    \textsc{Simple} \AcyclicTemporization is always \texttt{\emph{yes}} when the girth of the input digraph is at least 3, and \texttt{\emph{no}} otherwise.
    \label{thm:unbosimpleacy}
\end{lemma}

\begin{proof}
    By \Cref{lemma:lexicographicTemporizationJourneysAtMostLengthTwo}, the application of the \lexicographicTemporization ensures that any cycle of size at least 3 cannot be a \simplecycle, since a temporal path of length at least 3 is required. Hence, graphs of girth at least 3 can always be made acyclic through the \lexicographicTemporization. Concerning graphs of girth 2, we trivially note that no temporization can avoid a \simplecycle.  
 \end{proof}

\subsubsection{Lifetime at most 2.}

In the following, we prove that \textsc{Simple} \AcyclicTemporization becomes \NP-hard when the lifetime of the resulting temporal digraph is constrained to be at most 2. To this aim, we first observe the following property.

\begin{proposition}\label{prop:C4_label}
$\normalfont{(\star)}$ Let $\mathcal{D} = (D,\lambda)$ be a temporal digraph with lifetime 2. If $\mathcal{D}$ has no {\simplecycle}s, then $D$ has no cycles on less than 4 vertices. Additionally, if $C = (a,e_1,b,e_2,c,e_3,d,e_4,a)$ is a cycle in $D$, then $\lvert \lambda(e_i)\rvert = 1$ for every $i\in [4]$, and $\lambda(e_1) =\lambda(e_3) \neq \lambda(e_2) = \lambda(e_4)$. 
\end{proposition}
%STARHERE

We are now ready to construct our reduction. We reduce from \textsc{Monotone NAE 3-SAT}. Let $\phi$ be a formula on variables $x_1,\ldots,x_n$ and clauses $c_1,\ldots,c_m$. For each variable $x_i$, 
add to $D$ a $2\times (2m-1)$ grid as in Figure~\ref{fig:variable_gadget_onecycle}. Formally, add vertices $\{a^i_1,\ldots,a^i_{2m-1},b^i_1,\ldots,b^i_{2m-1}\}$, make the set $A^i = \{a^i_1,\ldots,a^i_{2m-1}\}$ form a non-oriented path from $a^i_1$ to $a^i_{2m-1}$ and $B^i = \{b^i_1,\ldots,b^i_{2m-1}\}$ form a non-oriented path from $b^i_1$ to $b^i_{2m-1}$. Then, orient such paths in a way that even vertices on the path formed by $A^i$ are sinks within the path, while even vertices in the path formed by $B^i$ are sources within the path. Finally, add the matching $\{b^i_{2j-1}a^i_{2j-1},a^i_{2j}b^i_{2j}\mid j\in [m]\}$ (in words, the odd arcs point from $b$ to $a$ and the even ones from $a$ to $b$). Denote by $D_i$ the gadget related to variable $x_i$. Observe that for each $j\in[m]$, we have a cycle $C^i_j = (a^i_{2j-1},a^i_{2j},b^i_{2j},b^i_{2j-1},a^i_{2j-1})$ on 4 vertices. By Proposition~\ref{prop:C4_label}, we know that all vertical arcs within $D_i$ are going to be given the same time, as well as all horizontal arcs. Hence, we use the vertical arcs to pass the value of $x_i$. This is formalized in the next paragraph in the construction of the clause gadgets.
%To help, let $c_j$ be a clause containing variable $x_i$. The \emph{active arc of $x_i$ for $c_j$} is defined as $b^i_{2j-1}a^i_{2j-1}$.
%, if this is a positive appearance; otherwise, it is defined as $a^i_{2j-1}b^i_{2j-1}$. 

Now, consider clause $c_j=(x_{i_1}\vee x_{i_2} \vee x_{i_3})$. We will link the appropriate arcs within the gadgets of $x_{i_1},x_{i_2},x_{i_3}$ in a way that they cannot all be assigned the same time. See Figure~\ref{fig:clause_gadget_onecycle} to follow the construction. First create a new vertex, $c_j$ in $D$. 
Then, identify vertices $a^{i_1}_{2j-1},b^{i_2}_{2j-1}$ and vertices $a^{i_2}_{2j-1},b^{i_3}_{2j-1}$; after this operation, $e_1,e_2,e_3$ form a path, denoted by $P_j$. Finally, add arcs $a^{i_3}_jc_j$ and $c_jb^{i_1}_j$. For each $j\in[2m-1]$, we denote by $C_j$ the set of vertices in ``column $j$'', i.e., $C_j = \{a_j^i,b_j^i\mid i\in[n]\}$
%$C_j = \{a_j^i,b_j^i\mid i\in[n]\setminus \{i_1,i_2,i_3\}\}\cup V(P_j)$.
Observe that $C_j$ induces a subgraph which is a matching together with path $P_j$ if $j$ is odd; otherwise $C_j$ induces a perfect matching from $a$'s to $b$'s. In either case, $C_j$ induces an acyclic subgraph.

\hspace{-0.6cm}\begin{minipage}{\linewidth}
      \centering
      \begin{minipage}{0.60\linewidth}
\begin{figure}[H]
%\input{figs/variable_gadgets}
% \begin{tikzpicture}[scale=1]
%   \pgfsetlinewidth{1pt}
%   \pgfdeclarelayer{bg}    
%    \pgfsetlayers{bg,main} 
  
%   \tikzset{vertex/.style={circle, minimum size=0.2cm, draw, inner sep=1pt, fill=black}}

%   \node[vertex,label=90:$a^i_1$] (a1) at (0,0) {};
%   \node[vertex,label=-90:$b^i_1$] (b1) at (0,-2) {};
%   \node[vertex,label=90:$a^i_2$] (a2) at (2,0) {};
%   \node[vertex,label=-90:$b^i_2$] (b2) at (2,-2) {};
%   \node[vertex,label=90:$a^i_3$] (a3) at (4,0) {};
%   \node[vertex,label=-90:$b^i_3$] (b3) at (4,-2) {};
%   \node[vertex,label=90:$a^i_4$] (a4) at (6,0) {};
%   \node[vertex,label=-90:$b^i_4$] (b4) at (6,-2) {};
%   \node[vertex,label=90:$a^i_5$] (a5) at (8,0) {};
%   \node[vertex,label=-90:$b^i_5$] (b5) at (8,-2) {};
  
%   \begin{pgfonlayer}{bg}    % select the background layer
%       %\draw (a1) edge [bend left] node [midway, fill=white] {$4$} (xi);
%       \draw[->, >=stealth] (a1)--(a2); \draw[<-, >=stealth] (a2)--(a3); \draw[->, >=stealth] (a3)--(a4); \draw[<-, >=stealth] (a4)--(a5);
%       \draw[<-, >=stealth] (b1)--(b2); \draw[->, >=stealth] (b2)--(b3); \draw[<-, >=stealth] (b3)--(b4); \draw[->, >=stealth] (b4)--(b5);
%       \draw[->, >=stealth] (b1)--(a1); \draw[->, >=stealth] (b3)--(a3); \draw[->, >=stealth] (b5)--(a5);
%       \draw[->, >=stealth] (a2)--(b2); \draw[->, >=stealth] (a4)--(b4);
%   \end{pgfonlayer}
%   \end{tikzpicture}
\includegraphics[scale=.85]{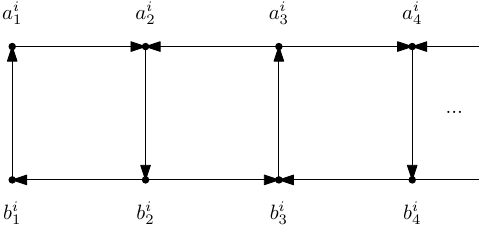}
\caption{Variable gadget in the reduction of \textsc{Simple} \AcyclicTemporization. There are a total of $2m-1$ $a^i_j$ vertices and $2m-1$ $b^i_j$ vertices per variable gadget.
\label{fig:variable_gadget_onecycle}}
\end{figure}
      \end{minipage}\quad\begin{minipage}{0.37\linewidth}
\begin{figure}[H]
%\input{figs/variable_gadgets}
% \begin{tikzpicture}[scale=1]
%   \pgfsetlinewidth{1pt}
%   \pgfdeclarelayer{bg}    
%    \pgfsetlayers{bg,main} 
  
%   \tikzset{vertex/.style={circle, minimum size=0.2cm, draw, inner sep=1pt, fill=black}}

%   \node[vertex,label=90:$b^{i_1}_1$] (w1) at (0,0) {};
%   %\node[vertex,label=-90:$c^1_1$] (c1) at (0,-2) {};
%   \node[vertex,label=-90:$c_1$] (c1) at (0,-2) {};
%   \node[vertex,label=45:$a^{i_1}_1$] (w2) at (2,0) {};
%   \node[vertex,label=-45:$a^{i_3}_1$] (y3) at (2,-2) {};
%   \node[vertex,label=0:$a^{i_2}_1$] (w3) at (3.5,-1) {};
%   %\node[vertex,label=180:$c_1^2$] (c2) at (-2,-2) {};
  
%   \begin{pgfonlayer}{bg}    % select the background layer
%       \draw[->, >=stealth] (w1) edge node [midway, fill=white] {$e_1$} (w2);
%       \draw[->, >=stealth] (w2) edge node [midway, fill=white] {$e_2$} (w3);
%       \draw[->, >=stealth] (w3) edge node [midway, fill=white] {$e_3$} (y3);
%       \draw[->, >=stealth] (c1)--(w1);
%       \draw[->, >=stealth] (y3)--(c1);
%       %\draw[<-] (c2) edge [bend left] (w1);
%       %\draw[<-] (y3) edge [out=-105,in=-75] (c2);
%   \end{pgfonlayer}
%   \end{tikzpicture}
\includegraphics[scale=.85]{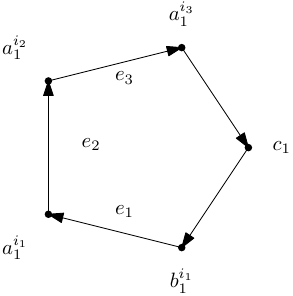}
\caption{Clause gadget in the reduction for \textsc{Simple} \AcyclicTemporization. The vertices of edges $e_j$ are identified with vertices of variable gadgets.
\label{fig:clause_gadget_onecycle}}
\end{figure}
      \end{minipage}
\vspace{0.1cm}
\end{minipage}

\begin{theorem}
$\normalfont{(\star)}$ Given a digraph $D$, deciding whether there exists a temporization $\lambda:E(D)\rightarrow 2^{[2]}$ such that $\mathcal{D} = (D,\lambda)$ contains no {\simplecycle}s is $\NP$-complete.
\label{thm:simpleacy}
\end{theorem}
%STARHERE

\subsection{Weak-cycles}
\label{sec:weak-cycles}

In the following, we prove some results for \textsc{Weak} \AcyclicTemporization depending on the girth of the input digraph. In particular, we show that whenever the girth is different from 4 we know an exact answer for our problem, while we leave open the case where the girth is 4.

\begin{lemma}
    $\normalfont{(\star)}$ \textsc{Weak} \AcyclicTemporization is always \texttt{\emph{yes}} when the girth of the input digraph is at least 5, and \texttt{\emph{no}} if girth is at most 3.
    \label{thm:unboweakacy}
\end{lemma}

\begin{proof}
    By \Cref{lemma:lexicographicTemporizationJourneysAtMostLengthTwo}, applying the \lexicographicTemporization ensures that any cycle of size at least 5 cannot be a \weakcycle, since for any pair $u,v$ of vertices in a cycle on $5$ vertices a temporal path of length at least 3 is required in order to have a \tpath{$v,u$} and a \tpath{$u,v$}. Hence, graphs of girth at least 5 can always be made acyclic through the \lexicographicTemporization. One can easily note that any temporization of a cycle on 3 vertices induces a temporal path of length at least 2, and thus contains a \weakcycle. Trivially, cycles on 2 vertices are \weakcycle. Therefore, no temporization on digraphs of girth at most 3 can avoid a \weakcycle.
 \end{proof}

\hspace{-0.3cm}\begin{minipage}{\linewidth}
      \centering
      \begin{minipage}{0.60\linewidth}
\begin{figure}[H]
\includegraphics[scale=.85]{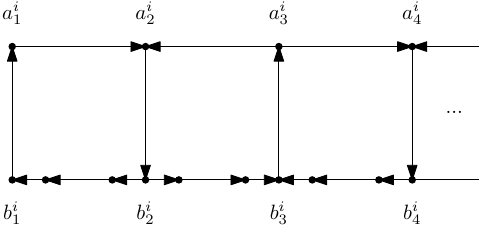}
\caption{Variable gadget in the reduction of \textsc{Weak} \AcyclicTemporization. There are a total of $2m-1$ $a^i_j$ vertices and $2m-1$ $b^i_j$ vertices per variable gadget.
\label{fig:variable_gadget_twocycle}}
\end{figure}
\end{minipage}\quad\begin{minipage}{0.37\linewidth}
\begin{figure}[H]
\begin{center}
%\input{figs/variable_gadgets}
% \begin{tikzpicture}[scale=1]
%   \pgfsetlinewidth{1pt}
%   \pgfdeclarelayer{bg}    
%    \pgfsetlayers{bg,main} 
  
%   \tikzset{vertex/.style={circle, minimum size=0.2cm, draw, inner sep=1pt, fill=black}}
%   \node[vertex,label=90:$b^{i_1}_1$] (w1) at (0,0) {};
%   %\node[vertex,label=-90:$c^1_1$] (c1) at (0,-2) {};
%   \node[vertex,label=45:$a^{i_1}_1$] (w2) at (2,0) {};
%   \node[vertex,label=-45:$a^{i_3}_1$] (y3) at (2,-2) {};
%   \node[vertex,label=0:$a^{i_2}_1$] (w3) at (3.5,-1) {};
%   %\node[vertex,label=180:$c_1^2$] (c2) at (-2,-2) {};
  
%   \node[vertex] (x) at (-2,0) {};
%   \node[vertex] (y) at (-2,-2) {};
%   \node[vertex] (w) at (0,-2) {};

%   \begin{pgfonlayer}{bg}    % select the background layer
%       \draw[->, >=stealth] (w1) edge node [midway, fill=white] {$e_1$} (w2);
%       \draw[->, >=stealth] (w2) edge node [midway, fill=white] {$e_2$} (w3);
%       \draw[->, >=stealth] (w3) edge node [midway, fill=white] {$e_3$} (y3);

%       \draw[->, >=stealth] (y3)--(w);
%       \draw[->, >=stealth] (w)--(y);
%       \draw[->, >=stealth] (y)--(x);
%       \draw[->, >=stealth] (x)--(w1);
%   \end{pgfonlayer}
%   \end{tikzpicture}
\includegraphics[scale=.85]{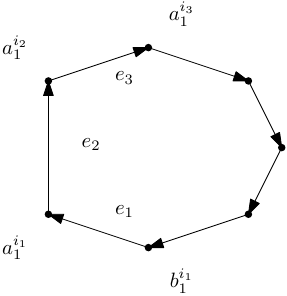}
\caption{Clause gadget in the reduction for \textsc{Weak} \AcyclicTemporization. The vertices of edges $e_j$ are identified with vertices of variable gadgets.
\label{fig:clause_gadget_twocycle}}
\end{center}
\end{figure}
% \begin{figure}[h]
% \begin{center}
% %\input{figs/variable_gadgets}
% \begin{tikzpicture}[scale=1]
%   \pgfsetlinewidth{1pt}
%   \pgfdeclarelayer{bg}    
%    \pgfsetlayers{bg,main} 
  
%   \tikzset{vertex/.style={circle, minimum size=0.2cm, draw, inner sep=1pt, fill=black}}
%   \node[vertex,label=90:$b^{i_1}_1$] (w1) at (0,0) {};
%   %\node[vertex,label=-90:$c^1_1$] (c1) at (0,-2) {};
%   \node[vertex,label=-90:$c_1$] (c1) at (0,-2) {};
%   \node[vertex,label=45:$a^{i_1}_1$] (w2) at (2,0) {};
%   \node[vertex,label=-45:$a^{i_3}_1$] (y3) at (2,-2) {};
%   \node[vertex,label=0:$a^{i_2}_1$] (w3) at (3.5,-1) {};
%   %\node[vertex,label=180:$c_1^2$] (c2) at (-2,-2) {};
%   \node[vertex] (x) at (-2,0) {};
%   \node[vertex] (y) at (-2,-2) {};
%   \node[vertex] (w) at (0,-3.5) {};
%   \begin{pgfonlayer}{bg}    % select the background layer
%       \draw[->] (w1) edge node [midway, fill=white] {$e_1$} (w2);
%       \draw[->] (w2) edge node [midway, fill=white] {$e_2$} (w3);
%       \draw[->] (w3) edge node [midway, fill=white] {$e_3$} (y3);
%       \draw[<-] (c1)--(w1);
%       \draw[<-] (y3)--(c1);

%       \draw[->] (y3)--(w);
%       \draw[->] (w)--(y);
%       \draw[->] (y)--(x);
%       \draw[->] (x)--(w1);
%   \end{pgfonlayer}
%   \end{tikzpicture}
% \caption{Clause gadget related to $c_1$ in the proof that deciding the existence of a temporization of $G$ that contains no {\weakcycle} is $\NP$-complete.}\label{fig:clause_gadget_twocycle}
% \end{center}
% \end{figure}
\end{minipage}
\end{minipage}

%\davi{discussion about girth 4}\davi{2: it is already in the conclusion. should we do a brief comment here?}

\subsubsection{Lifetime at most~2.}
We present a reduction showing that \textsc{Weak} \AcyclicTemporization is \NP-complete if the lifetime is constrained to be $2$. We make use of the following proposition.

\begin{proposition}\label{prop:C6_label}
$\normalfont{(\star)}$ Let $\mathcal{D} = (D,\lambda)$ be a temporal digraph with lifetime 2. If $\mathcal{D}$ has no {\weakcycle}s, then $D$ has no cycles on less than 6 vertices. Additionally, if $C = (a,e_1,b,e_2,c,e_3,d,e_4,e, e_5,f,e_6,a)$ is a cycle in $D$, then $\lvert \lambda(e_i)\rvert = 1$ for every $i\in [6]$, and $\lambda(e_1) =\lambda(e_3)=\lambda(e_5) \neq \lambda(e_2) = \lambda(e_4) = \lambda(e_6)$. 
\end{proposition}
%STARHERE

We now use a similar reduction to the one presented in Section~\ref{sub:simple} to prove that it is hard to construct a temporization with no {\weakcycle}s that uses only times within $2^{[2]}$. One can see the construction of the variable gadget $D_i$ as obtained from the previous construction by replacing each arc $dd'$ on the lower path by a directed path on three arcs (see Figure~\ref{fig:variable_gadget_twocycle}). Formally, add vertices $\{a^i_1,\ldots,a^i_{2m-1},b^i_1,\ldots,b^i_{2m-1}\}$ and the matching $\{b^i_{2j-1}a^i_{2j-1},a^i_{2j}b^i_{2j}\mid j\in [m]\}$ (in words, the odd arcs point from $b$ to $a$ and the even ones from $a$ to $b$). Denote the set $\{a^i_1,\ldots,a^i_{2m-1}\}$ by $A^i$ and the set $\{b^i_1,\ldots,b^i_{2m-1}\}$ by $B^i$; also let $A = \bigcup_{i\in[n]}A^i$ and $B = \bigcup_{i\in [n]}B^i$. 
Then, for each $i\in[n]$ make the set $A^i$ form a non-oriented path from $a^i_1$ to $a^i_{2m-1}$ and orient such path in a way that even vertices on the path formed by $A^i$ are sinks within the path. 
Finally, for each $j\in [m-1]$, create a path on three arcs from $b^i_{2j}$ to $b^i_{2j-1}$ and a path on three arcs from $b^i_{2j}$ to $b^i_{2j+1}$, adding four new vertices in order to do so. 
Again, denote by $D_i$ the gadget related to variable $x_i$ and observe that for each $j\in[m-1]$, we have a cycle on 6 vertices containing arcs $b^i_{2j-1}a^i_{2j-1}$ and $a^i_{2j}b^i_{2j}$ and that the latter arc is also within a cycle on 6 vertices with arc $b^i_{2j+1}a^i_{2j+1}$. As before, we call the arcs between $A$ and $B$ \emph{vertical arcs}. By Proposition~\ref{prop:C6_label}, we know that all vertical arcs within $D_i$ are going to be given the same time and hence can be used to pass the value of $x_i$. 

Now consider the clause $c_j = (x_{i_1}\vee x_{i_2} \vee x_{i_3})$. As before, we identify vertices in column $2j-1$ in order to make the vertical arcs of these variables form a path. 
% Additionally, we want to create a $C_6$ intersecting such path exactly on the endpoints (see Figure~\ref{fig:clause_gadget_twocycle}). Formally, first add vertex $c_j$ and edges $b^{i_1}_{2j-1}c_j$ and $c_ja^{i_3}_{2j-1}$. Then, add two new vertices and 4 arcs in order to form a cycle on 6 vertices containing the path $(b^{i_1}_{2j-1},c_j,a^{i_3}_{2j-1})$. We are now ready to prove hardness.
Additionally, we want to create a $C_7$ containing such path (see Figure~\ref{fig:clause_gadget_twocycle}). Formally, add three new vertices and 4 arcs in order to form a cycle on 7 vertices containing the path $(b^{i_1}_{2j-1},a^{i_1}_{2j-1},a^{i_2}_{2j-1},a^{i_3}_{2j-1})$. This results in the following.

\begin{theorem}
$\normalfont{(\star)}$
Given a digraph $D$, deciding whether there exists a temporization $\lambda:E(D)\rightarrow 2^{[2]}$ such that $\mathcal{D} = (D,\lambda)$ contains no {\weakcycle}s is $\NP$-complete.
\label{thm:weakacy}
\end{theorem}
%STARHERE

%%%%%%%%%%%%%%%%%%%%%%%%%%%%%%%%%%%%%%%%%%%%%%%%%%%%%%%%%%%%%%%%%%%%%%%%%%%%%%%%%%%%%%%%%%%%%%%%%%%%%%%%%%%%%%%%%%%%%%%%%%%%%%%%%%%%%%%%%%%%%%%
%%%%%%%%%%%%%%%%%%%%%%%%%%%%%%%%%%%%%%%%%%%%%%%%%%%%%%%%%%%%%%%%%%%%%%%%%%%%%%%%%%%%%%%%%%%%%%%%%%%%%%%%%%%%%%%%%%%%%%%%%%%%%%%%%%%%%%%%%%%%%%%
%%%%%%%%%%%%%%%%%%%%%%%%%%%%%%%%%%%%%%%%%%%%%%%%%%%%%%%%%%%%%%%%%%%%%%%%%%%%%%%%%%%%%%%%%%%%%%%%%%%%%%%%%%%%%%%%%%%%%%%%%%%%%%%%%%%%%%%%%%%%%%%

%\subsection{The undirected case}
%\red{Ana: I think we should concentrate only on the directed case}
%\section{Acyclic temporization for case 2}

% \subsection{tau bounded}
% The problem is NP-complete reducing from NAE.\todo{am}{The one with exagons}

% \subsection{tau unbounded}
% The idea here is to use the same strategy of ordering the edges above but ensuring an ordering $\sigma$ of the vertices without rainbows for $C_4$.

% In particular, with an arbitrary order $\sigma$, we already know that the ordering is taking care of all the cycles of length at least 5.
% Moreover if there are cycles directed of length 3 we can answer directly NO. Hence, we know that if the girth is 3 then the answer is NO, if the girth is at least 5, the answer is always YES.

% Trying to find ordering of the vertices avoiding rainbows in $C_4$s will make work the labeling also for all the $C_4$s.
% Our conjecture is that for every graph such ordering exists.
\section{Conclusion}

We present multiple manners of extending the concept of a cycle in temporal graphs. For each of these, interesting problems are studied in terms of algorithmics and tractability (see again \Cref{fig:results}).

Throughout the paper, we considered non-strict temporal paths. When we consider strict temporal paths instead, i.e. paths that can only use strictly increasing times, \AcyclicTemporization becomes trivial: assign time 1 to each arc and all cycles are avoided (except for \weakcycles of size 2 which cannot be avoided in any manner). For \CycleDetection, we claim that all algorithms presented can be adapted easily for strict temporal paths. These algorithms use the computation of EAT as a subroutine, and so a quick modification to how these corresponding searches extend is sufficient. Finally, since our reduction for \textsc{Strong} \CycleDetection uses a proper time function (in other words, there are no two arcs with the same time that can concatenate to form a non-strict temporal path), it holds for the strict case as well. 

%\Jason{small discussion about undirected graphs?}
All problems in this paper have been solved (in terms of tractability or feasability), sometimes by covering the different cases such as when the answer of \AcyclicTemporization depends on the girth of the digraph. One case remains stubborn however: \textsc{Weak} \AcyclicTemporization in the case of girth 4. One manner to solve this case in the affirmative, is to prove that for any such a digraph, an ordering of the vertices exists such that \lexicographicTemporization according to this order would result in an acyclic temporization. Furthermore, we have been unable to obtain an example digraph of girth 4 in which a \weakcycle cannot be avoided, which may support the idea that an acyclic temporization always exists. We leave this as an open problem. 

\bibliographystyle{splncs04}
\bibliography{paper} 

\begin{thebibliography}{10}
\providecommand{\url}[1]{\texttt{#1}}
\providecommand{\urlprefix}{URL }
\providecommand{\doi}[1]{https://doi.org/#1}

\bibitem{AkridaMSR21}
Akrida, E.C., Mertzios, G.B., Spirakis, P.G., Raptopoulos, C.L.: The temporal
  explorer who returns to the base. J. Comput. Syst. Sci.  \textbf{120},
  179--193 (2021)

\bibitem{Bang-JensenBGP22}
Bang{-}Jensen, J., Bessy, S., Gon{\c{c}}alves, D., Picasarri{-}Arrieta, L.:
  Complexity of some arc-partition problems for digraphs. Theor. Comput. Sci.
  \textbf{928},  167--182 (2022). \doi{10.1016/J.TCS.2022.06.023},
  \url{https://doi.org/10.1016/j.tcs.2022.06.023}

\bibitem{Bartlang.book}
Bartlang, U.: Architecture and methods for flexible content management in
  peer-to-peer systems. Springer (2010)

\bibitem{BCP13}
Bermond, J.C., Cosnard, M., Pérennes, S.: Directed acyclic graphs with the
  unique dipath property. Theoretical Computer Science  \textbf{504},  5--11
  (2013). \doi{https://doi.org/10.1016/j.tcs.2012.06.015},
  \url{https://www.sciencedirect.com/science/article/pii/S0304397512005841},
  discrete Mathematical Structures: From Dynamics to Complexity

\bibitem{BumpusM23}
Bumpus, B.M., Meeks, K.: Edge exploration of temporal graphs. Algorithmica
  \textbf{85}(3),  688--716 (2023)

\bibitem{Cormen.book}
Cormen, T.H., Leiserson, C.E., Rivest, R.L., Stein, C.: Introduction to
  algorithms. MIT press (2022)

\bibitem{Erlebach0K21}
Erlebach, T., Hoffmann, M., Kammer, F.: On temporal graph exploration. J.
  Comput. Syst. Sci.  \textbf{119},  1--18 (2021)

\bibitem{ErlebachS23}
Erlebach, T., Spooner, J.T.: Parameterised temporal exploration problems. J.
  Comput. Syst. Sci.  \textbf{135},  73--88 (2023)

\bibitem{FHW.80}
Fortune, S., Hopcroft, J., Wyllie, J.: The directed subgraph homeomorphism
  problem. Theoretical Computer Science  \textbf{10}(2),  111--121 (1980)

\bibitem{Getal.14}
Ganian, R., Hlin{\v{e}}n{\`y}, P., Kneis, J., Langer, A., Obdr{\v{z}}{\'a}lek,
  J., Rossmanith, P.: On digraph width measures in parameterized algorithmics.
  Discrete Applied Mathematics  \textbf{168},  88--107 (2014)

\bibitem{HaagMNR22}
Haag, R., Molter, H., Niedermeier, R., Renken, M.: Feedback edge sets in
  temporal graphs. Discret. Appl. Math.  \textbf{307},  65--78 (2022).
  \doi{10.1016/J.DAM.2021.09.029},
  \url{https://doi.org/10.1016/j.dam.2021.09.029}

\bibitem{Holme.15}
Holme, P.: Modern temporal network theory: a colloquium. The European Physical
  Journal B  \textbf{88}(9), ~234 (2015)

\bibitem{JRST.82}
Johnson, T., Robertson, N., Seymour, P.D., Thomas, R.: Directed tree-width.
  Journal of Combinatorial Theory, Series B  \textbf{82}(1),  138--154 (2001)

\bibitem{klobas_et_al:LIPIcs.MFCS.2022.62}
Klobas, N., Mertzios, G.B., Molter, H., Spirakis, P.G.: {The Complexity of
  Computing Optimum Labelings for Temporal Connectivity}. In: Szeider, S.,
  Ganian, R., Silva, A. (eds.) 47th International Symposium on Mathematical
  Foundations of Computer Science (MFCS 2022). Leibniz International
  Proceedings in Informatics (LIPIcs), vol.~241, pp. 62:1--62:15. Schloss
  Dagstuhl -- Leibniz-Zentrum f{\"u}r Informatik, Dagstuhl, Germany (2022).
  \doi{10.4230/LIPIcs.MFCS.2022.62},
  \url{https://drops.dagstuhl.de/entities/document/10.4230/LIPIcs.MFCS.2022.62}

\bibitem{klobas_et_al:LIPIcs.SAND.2024.16}
Klobas, N., Mertzios, G.B., Molter, H., Spirakis, P.G.: {Temporal Graph
  Realization from Fastest Paths}. In: Casteigts, A., Kuhn, F. (eds.) 3rd
  Symposium on Algorithmic Foundations of Dynamic Networks (SAND 2024). Leibniz
  International Proceedings in Informatics (LIPIcs), vol.~292, pp. 16:1--16:18.
  Schloss Dagstuhl -- Leibniz-Zentrum f{\"u}r Informatik, Dagstuhl, Germany
  (2024). \doi{10.4230/LIPIcs.SAND.2024.16}

\bibitem{LVM.18}
Latapy, M., Viard, T., Magnien, C.: Stream graphs and link streams for the
  modeling of interactions over time. Social Network Analysis and Mining
  \textbf{8}(1), ~61 (2018)

\bibitem{MarinoS23}
Marino, A., Silva, A.: Eulerian walks in temporal graphs. Algorithmica
  \textbf{85}(3),  805--830 (2023)

\bibitem{mertzios2024realizingtemporaltransportationtrees}
Mertzios, G.B., Molter, H., Spirakis, P.G.: Realizing temporal transportation
  trees (2024), \url{https://arxiv.org/abs/2403.18513}

\bibitem{M.16}
Michail, O.: An introduction to temporal graphs: An algorithmic perspective.
  Internet Mathematics  \textbf{12}(4),  239--280 (2016)

\bibitem{Netal.13}
Nicosia, V., Tang, J., Mascolo, C., Musolesi, M., Russo, G., Latora, V.: Graph
  metrics for temporal networks. In: Temporal networks, pp. 15--40. Springer
  (2013)

\bibitem{Sapatnekar.book}
Sapatnekar, S.: Timing. Springer (2004)

\bibitem{shmulevich.book}
Shmulevich, I., Dougherty, E.R.: Probabilistic Boolean networks: the modeling
  and control of gene regulatory networks. SIAM (2010)

\bibitem{wu2014path}
Wu, H., Cheng, J., Huang, S., Ke, Y., Lu, Y., Xu, Y.: Path problems in temporal
  graphs. Proceedings of the VLDB Endowment  \textbf{7}(9),  721--732 (2014)

\bibitem{xuan2003computing}
Xuan, B.B., Ferreira, A., Jarry, A.: Computing shortest, fastest, and foremost
  journeys in dynamic networks. International Journal of Foundations of
  Computer Science  \textbf{14}(02),  267--285 (2003)

\end{thebibliography}

\appendix

\section{Omitted Proofs for \CycleDetection}

\subsection{Proof of~\Cref{prop:weak_cycle_detection_poly}}
\begin{proof}
In order to detect a {\weakcycle}, it suffices to test for every pair $x,y\in V(\mathcal{D})$, whether $x$ reaches $y$ and $y$ reaches $x$. This means determining whether the earliest arrival times between these vertices are finite. Indeed, suppose this is the case and let $P_{xy} = (v_1 = x,t_1,\ldots,t_p,v_{p+1}=y)$ (resp. $P_{yx} = (v'_1 = y,t'_1,\ldots,t'_q,v'_{q+1}=y)$) be a temporal path in $\mathcal{D}$ from $x$ to $y$ (resp. from $y$ to $x$). If $P_{xy}$ and $P_{yx}$ intersect only in $x$ and $y$, we are done. So suppose that they intersect in at least one other vertex and let $i$ be minimum such that $v_i\in (V(P_{xy})\cap V(P_{yx}))\setminus\{x,y\}$.  Now, let $Q$ be the temporal $x,v_i$-path contained in $P_{xy}$ and $Q'$ be the temporal $v_i,y$-path contained in $P_{yx}$. By the choice of $v_i$, observe that $Q'$ cannot intersect $Q$ on an internal vertex. It follows that the concatenation of $Q$ and $Q'$ forms a {\weakcycle} in $\mathcal{D}$.
 \end{proof}

\subsection{Proof of~\Cref{prop:simple_cycle_detection_poly}}
\begin{proof}
To detect a \simplecycle that starts in a vertex $v$ in a given temporal digraph $\mathcal{D} = (D, \lambda)$, we go over each arc %$(v, r) \in A(D)$,  changed arc
$vr \in A(D)$, and check whether EAT$(r, v) \leq \max(\lambda(v, r))$. If this holds, then a \simplecycle exists, formed by the temporal path from $r$ to $v$ concatenated with arc %$(v, r)$. changed arc  
$vr$. If no vertex $r$ can reach any of their incoming neighbors $v$ in time before using the arc %$(v, r)$ changed arc
$vr$, then no \simplecycle exists. Thus, it suffices to repeat the above procedure for every $v\in V(\mathcal{D})$.
 \end{proof}

\subsection{Proof of~\Cref{prop:aux_cycle_unicity_and_disjointness}}
\begin{proof}
    Let $\mathcal{A}$ be an auxiliary cycle of order $n$. First, note that given $v_i\in V(\mathcal{A})$ such that $i \neq n-1$,%$v_i \neq v_{n-1} \in V(\mathcal{A})$\julio{$v_i\in V(\mathcal{A})$ such that $i \neq n-1$}
    $W_{v_i} = (v_i, t_i^1, v_{i+1}, t_i^2, \dots, v_{n-1}, t_i^{n-i}, v_0, t_i^{n-i+1}, \dots, t_i^n, v_i$), where $t_i^j = j(n-1) -i$, is a \tpath{$v_i,v_i$}. 
    % Furthermore, for $v_{n-1}$, we have the \tpath{$x,x$} $W_x = (x, 0, t_1^1, v_1, t_1^2, \dots, t_1^{n-1},x)$. Thus, for each vertex $v \in V(\mathcal{A})$, there is a \tpath{$v,v$}.

    We now prove that $W_{v_i}$ is the only \tpath{$v_i,v_i$} for each vertex $v_i$. Notice that, in $W_{v_i}$, each time $t^j_i$ is the minimum possible that maintains the temporal path. Indeed, for any $j>1$, the times smaller than $t_i^j$ on the same arc are at most $t_i^j - n =(j-1)(n-1)-i-1$, which is smaller than $t_i^{j-1} = (j-1)(n-1)-i$. Moreover, if we take a time greater than $t_i^j$ on the same arc, the last time must be at least $t_i^n + n = n^2 - i$, by construction. Note that the last arc in the \tpath{$v_i,v_i$} is $e_i$, whose greatest time is equal to $(n-1)n-i$, which is smaller than $n^2 - i$. Therefore, it is not possible to take a time greater than $t_i^j$ for any $j$. $W_{v_i}$ is thus the only \tpath{$v_i, v_i$}.

    Additionally, we prove that $W_{v_i}$ and $W_{v_k}$ are disjoint, for $i \neq k$. Since $i \neq k$, it follows that $t_i^j \neq t_k^j$, and thus the times are all different.
 \end{proof}

\subsection{Proof of~\Cref{theorem:strong_cycle_detection_NPcomplete}}
\begin{proof}
    \textsc{Strong} \CycleDetection is in \NP, because a solution subgraph $\mathcal{C}$ can be verified to be a cycle in the underlying graph, and deciding whether each vertex reaches itself can be done by checking whether EAT$(v, u)$ in $\mathcal{C}$ is at most $\max(\lambda(u,v))$, for each arc $%(u,v) changed arc
    uv \in A(\mathcal{C})$, similarly to \Cref{prop:weak_cycle_detection_poly}.
	
	To prove this problem is \NP-hard, we reduce \ThreeSAT to it. Let the generic instance of \ThreeSAT be the CNF formula $\phi$ of $n$ variables $x_0, x_1, ..., x_{n-1}$ and $m$ clauses $C_0, C_1, ..., C_{m-1}$. Let the literals of clause $C_i$ be denoted as $\ell_{i,1}, \ell_{i,2}$, and $\ell_{i,3}$. 
	Let us build an instance of \CycleDetection as the temporal digraph $\mathcal{D}(\phi)$ as follows. Initially, add three auxiliary cycles $\mathcal{A}^1$, $\mathcal{A}^2$, and $\mathcal{A}^3$, all of order $4m+1$. Let the corresponding vertices of $\mathcal{A}^1$, $\mathcal{A}^2$ and $\mathcal{A}^3$ be referred to as $v^1_i$, $v^2_i$, and $v^3_i$ respectively, for every $i\in\{0,\ldots, 4m\}$. Note that for any three vertices $v^1_i$, $v^2_i$, and $v^3_i$,  $L^\circlearrowright(v^1_i) = L^\circlearrowright(v^2_i) = L^\circlearrowright(v^3_i)$ Let us simply refer to these times as $L^\circlearrowright(v_i)$ instead.
	Now, for each $i \in \mathbb{N}$, merge\footnote{We define merging of vertices in temporal digraphs as in static digraphs, and times on arcs of pre-merged vertices remain on corresponding arcs of post-merged vertices.} the three vertices $v^1_{4i}$, $v^2_{4i}$, and $v^3_{4i}$, and refer to this merged vertex as $v_{4i}$ (see \Cref{fig:reduction_strongcycledetection}). Note that this also merges vertices $v^1_{4m}$, $v^2_{4m}$, and $v^3_{4m}$ into vertex $v_{4m}$.
	For each clause $C_i$, add time 0 to arcs %$(v^1_{4i+2}, v^1_{4i+3})$, $(v^1_{4i+3}, v_{4(i+1)})$, $(v^2_{4i+1}, v^2_{4i+2})$, $(v^2_{4i+3}, v_{4(i+1)})$, $(v^3_{4i+1}, v^3_{4i+2})$, and $(v^3_{4i+2}, v^3_{4i+3})$. changed arc
    $v^1_{4i+2}v^1_{4i+3}$, $v^1_{4i+3}v_{4(i+1)}$, $v^2_{4i+1}v^2_{4i+2}$, $v^2_{4i+3}v_{4(i+1)}$, $v^3_{4i+1}v^3_{4i+2}$, and $v^3_{4i+2}v^3_{4i+3}$.
	Finally, for each literal $\ell_{i, j}$ corresponding to some variable $x_k$, and each literal $\ell_{g, h}$ which corresponds to $\neg x_k$, remove from arc %$(v_{4g}, v^h_{4g+1})$ changed arc
    $v_{4g}v^h_{4g+1}$ all times from $L^{\circlearrowright}(v_{4i+j})$. This concludes the transformation. 
    
	Let us now prove that a positive instance for \ThreeSAT remains a positive instance for \CycleDetection after the transformation (\ThreeSAT $\implies$ \CycleDetection), and vice versa, that a positive instance of \CycleDetection in the transformed instance implies a positive instance of \ThreeSAT before the transformation (\CycleDetection $\implies$ \ThreeSAT).  The key idea for both directions is that a literal $\ell_{i, j}$ is true, if and only if path $(v^j_{4i+1}, v^j_{4i+2}, v^j_{4i+3})$ is part of a \strongcycle.
	
	\ThreeSAT $\implies$ \CycleDetection:\\
	It is clear that all vertices of the form $v_{4i}$ (i.e. all merged vertices, or all vertices without superscript) must be part of any solution of the transformed \CycleDetection instance, as no \strongcycle exists which doesn't use them. Let us select these vertices as an incomplete solution $S'$ for the \CycleDetection instance. Suppose a solution $S$ for the \ThreeSAT formula $\phi$. For each clause $C_i$, find a literal which is true according to $S$ (there must be at least one since $S$ is a solution). Suppose it is literal $\ell_{i,j}$. Add to $S'$ the vertices $v^j_{4i+1}$, $v^j_{4i+2}$, and $v^j_{4i+3}$. After having done this for each clause, we now claim that $S'$ (or rather, the induced temporal subgraph $\mathcal{D}[S']$) is a solution for the transformed \CycleDetection instance. 
	It should be clear that $S'$ induces a cycle. All that remains is to prove that all vertices on this cycle can reach themselves:
	\begin{itemize}
		\item Vertex $v_{4m}$: vertices $v^k_{4m}$ by definition have time 0 on the outgoing arcs, which vertex $v_{4m}$ retains after the transformation, on arc %$(v_{4m}, v_0)$ changed arc
        $v_{4m}v_0$. This implies that $v_{4m}$ can reach $v_0$ before the latter reaches anything. By construction, if $v_0$ can reach itself in $S'$, it must do so by passing through $v_{4m}$, meaning $v_{4m}$ could reach itself as well. The bullet point below shows how $v_0$ reaches itself in $S'$. 
		\item Vertices of the form $v_{4i}$: by definition all $v^k_{4i}$ could reach themselves in $\mathcal{A}^k$ through times $L^\circlearrowright(v_{4i})$. Since the transformation only merged vertices and only removed times which are not part of $L^\circlearrowright(v_{4i})$, the self-reachability of $v^k_{4i}$ and thus of $v_{4i}$ was not altered by the transformation. What's more, since one of the three outgoing paths from any $v_{4j}$ to any $v_{4(j+1)}$ is included in $S'$, all $v_{4i}$ can still reach themselves in $S'$ since all three paths are identical in terms of times $L^\circlearrowright(v_{4i})$.
		\item Vertices of the form $v^j_{4i+j}$: if these vertices are part of $S'$, then it implies that literal $\ell_{i, j}$ is set to true in $S$. 
		Note first that these vertices used times $L^\circlearrowright(v_{4i+j})$ to reach themselves in $\mathcal{A}^j$. The vertex merging in the transformation did not alter the self-reachability of $v^j_{4i+j}$. What's more, since $S$ evaluates $\ell_{i,j}$ to true, it does not evaluate any contradicting literal $\ell_{g,h} = \neg \ell_{i, j}$ to true. This means that no arc $(v_{4g}, v^h_{4g+1})$ is part of $\mathcal{D}[S']$, and so no times of $L^\circlearrowright(v_{4i+j})$ have been removed in $\mathcal{D}[S']$, which ensures it can reach itself still.
		\item Other vertices in $S'$: by construction, only vertices of the form $v^k_{4i+j}$ such that $k \neq j$ and $v^j_{4i+j} \in S'$, are concerned by this bullet point. Similarly as to $v_{4m}$, our transformation allows these vertices to reach either $v^j_{4i+j}$ or $v_{4(i+1)}$ through times 0. We have proven that the latter vertices are able to reach themselves in $S'$ in the above bullet points, which by construction means that the former can reach themselves as well.
	\end{itemize}
	
	\CycleDetection $\implies$ \ThreeSAT:\\
	Suppose a solution of the transformed \CycleDetection instance to be \strongcycle $\mathcal{D}[S']$. Note again that this solution must include all vertices $v_{4i}$, since no cycle exists without these vertices. Note also that, due to $S'$ inducing a cycle, if some vertex $v^k_{4i+j}$ is in $S'$, then automatically all vertices between $v_{4i}$ and $v_{4(i+1)}$ that have superscript $k$ must also be in $S'$, and the vertices that have superscript $\neq k$ cannot be in $S'$.
	Let the solution $S$ for the \ThreeSAT instance be constructed as follows. For each clause $C_i$, determine what superscript the vertices between vertex $v_{4i}$ and vertex $v_{4(i+1)}$ have. Suppose it is $k$. Set the literal $\ell_{i, k}$ to true. After having done this for each clause, we now claim that $S$ is a solution for the initial \ThreeSAT instance. If some variable has not been assigned true or false, set it to false.  
	It should be clear that all clauses evaluate to true with $S$, since we constructed $S$ by assigning true to one of the literals of each clause. What remains to be proven is that the assignment $S$ does not assign true to both some variable $x_i$ and the negation $\neg x_i$. 
    By construction of $\mathcal{D}$, and by the structure pointed out in \Cref{prop:aux_cycle_unicity_and_disjointness}, we know that no vertices $v^j_{4i+j}$ and $v^h_{4g+h}$ can be part of $S'$ if literals $\ell_{i, j} = \neg \ell_{g, h}$, since otherwise $v^j_{4i+j}$ cannot reach itself as it must use arc %$(v_{4g}, v^h_{4g+1})$ changed arc
    $v_{4g}v^h_{4g+1}$ on which times $L^\circlearrowright(v_{4i+j})$ are missing. No contradicting literals can thus be assigned true in the construction of $S$, as the corresponding vertices can not both be part of $S'$. 
 \end{proof}

\subsection{Fixed-parameter tractability wrt lifetime: 
Proof of~\Cref{theorem:strong_cycle_detection_fpt}: 
}

Here we give the details and proofs corresponding to the fixed-parameter tractability section concerning \textsc{Strong} \CycleDetection. Let us start with a pseudo-code implementation. We choose a list structure for the search path $P$, and besides the typical list and array operations, function $A$ returns the arcs adjacent to a given vertex, and the other undefined functions (e.g. \texttt{earliestAtLeast}) should be clear from the context and function name. 

\begin{algorithm}[]
\DontPrintSemicolon
\caption{\texttt{containsStrongCycle}}\label{algo:full_cycle_detection}
\SetKwInOut{Input}{Input}
\SetKwInOut{Output}{Output}

\Input{temporal digraph $\mathcal{D}$}
\Output{\texttt{true} if $\mathcal{D}$ contains a \strongcycle, \texttt{false} otherwise}	
\For{%$a = (r,v)$ changed arc
	$a = rv$ \texttt{in} $A(\mathcal{D})$}{ 
	$P \gets $ \texttt{newList}$(r, v)$ \tcp*{search path}
	$T_r \gets$ $[0] * (\tau + 1)$ \tcp*{root timetable} 
	\For{$t \in [0, ..., \max(\lambda(a))]$}{
		$T_r[t] \gets \texttt{earliestAtLeast}(\lambda(a), t)$\tcp*[f]{EAT$_P^{\geq t}(r,v)$} 
	}
	$T_P \gets$ $[0] * (\tau + 1)$ \tcp*{path timetable} 
	$T_P[\tau] \gets \texttt{min}(\lambda(a))$ \tcp*[]{LDT$_P(r,r)$ and EAT$_P(r,v)$} 
	% $T_P[] \gets \texttt{min}(\lambda(a))$ \tcp*{$[$EAT$(r,v),$ LDT$(r, r)]$}
	% \texttt{add}($W_p$, $[0, \texttt{max}(\lambda(a))]$) \tcp*{$[$EAT$(v,v),$ LDT$(r, v)]$}
	% $T \gets \texttt{newStack}((T_r, T_P))$ \tcp*{timetables stack}
	$B \gets ([\texttt{false}] * m) * 4^{(\tau+1)^2}$ \tcp*{blocking matrix}
	\lIf{$\texttt{searchStrongCycle}(\mathcal{D}, P, T_r, T_P, B)$}
	{
		\Return \texttt{true}
	}
}
\Return \texttt{false}\;
\end{algorithm}	

\begin{algorithm}[]
\DontPrintSemicolon
\caption{\texttt{searchStrongCycle} (see also \Cref{fig:strong_cycle_search})}\label{algo:search_full_cycle}
\SetKwInOut{Input}{Input}
\SetKwInOut{Output}{Output}

\Input{temporal digraph $\mathcal{D}$, search path $P$, root and path timetables $T_r$ and $T_P$, and blocking matrix $B$}
\Output{\texttt{true} iff the current search path $P$ results in a strong-cycle}

%$(T_r, T_P) \gets \texttt{peek}(T)$\;
$u \gets \texttt{last}(P)$\;
\For{%$a = (u, v)$ changed arc
	$a = uv$ \texttt{in} $A(u)$}{ 
	\If(\tcp*[f]{$a$ doesn't block $T_r$ and $T_P$}){\texttt{not} $B[\texttt{order}(a)][\texttt{order}(T_r, T_P)]$}{
		% $i_u \gets \texttt{maxIndexWithNonZeroValue}(T_r)$\;
		$(T'_{r}, T'_{P},$ extended) $\gets \texttt{extend}(T_r, T_p, \lambda(a))$ \;
		\If{$extended$}
		{
			\lIf{$v \in P$}{\Return{\texttt{true}}}\label{line:returnTrue}
			% $w_v \gets [0, \texttt{max}(W_{r+1}, \texttt{key}=\texttt{lambda }[a,b] : a)]$ \tcp*{$[0$, LDT$(r, v)]$}
			% $T'_P[$index$] \gets \max(T'_P[$index$], \min(\lambda(a)))$\;
			% $w_v \gets [0, \texttt{max}(t \texttt{ for } [t, \_] \in W_{r+1})]$ \tcp*{$[$EAT$(v, v)$, LDT$(r, v)]$}
			% % $w_v \gets [0, \texttt{LDT}(P, r, v)]$\;
			% \texttt{add}$(W_{p+1}, w_v)$\; 
%			\texttt{add}$(T, (T'_{r}, T'_{P}))$\;	
			\texttt{add}$(P, v)$\;
			\lIf{\texttt{searchStrongCycle}$(\mathcal{D}, P, T'_r, T'_P, B)$}
			{
				\Return \texttt{true} 
			}
%			\texttt{pop}$(T)$\;
			\texttt{remove}$(P, v)$\;
		}
		$B[\texttt{order}(a)][\texttt{order}(T_r, T_P)] \gets \texttt{true}$ \tcp*[f]{block $T_r$ and $T_P$ on $a$}
	}
}
\Return \texttt{false}\;
\end{algorithm}	

\begin{algorithm}[]
\DontPrintSemicolon
\caption{\texttt{extend}}\label{algo:extend}
\SetKwInOut{Input}{Input}
\SetKwInOut{Output}{Output}

\Input{Root timetable $T_r$, path timetable $T_P$, and arc $a = uv$}
\Output{Root timetable $T'_{r}$ and path timetable $T'_{P}$, corresponding to the input timetables extended over $a$, and a boolean value that flags \texttt{false} if the path timetable cannot extend correctly}

$T'_{r} \gets \texttt{copy}(T_r)$\;
$T'_{P} \gets \texttt{copy}(T_P)$\;
\For{$t$ \texttt{in} $[1, ..., \tau]$}{ 
	\If{$T_r[t] > 0$ \texttt{and} \texttt{containsAtLeast}$(\lambda(a), T_r[t]+1)$}
	{
		$T'_r[t] \gets \texttt{earliestAtLeast}(\lambda(a), T_r[t]+1)$ \tcp*[f]{update EAT$_P^{\geq t}(r, v)$}
	}
	\lElse(){
		$T'_r[t] \gets 0$} 
	\If{$T_P[t] > 0$ \texttt{and} \texttt{containsBetween}$(\lambda(a), T_P[t], t+1)$}
	{
		$T'_P[t] \gets \texttt{earliestBetween}(\lambda(a), T_P[t], t+1)$\tcp*[f]{update EAT$_P(\_,v)$}
	}
	\lElse{\Return ($T_r, T_P, \texttt{false}$)}
}
$i \gets \texttt{maxIndexWithNonZeroValue}(T_r)$\tcp*[]{LDT$_P(r,u)$}
$T'_{P}[i-1] \gets \max(T'_{P}[i-1], \min(\lambda(a)))$\tcp*[]{add EAT$_P(u,v)$}
\Return $(T'_{r}, T'_{P}, \texttt{true})$\;
\end{algorithm}	

\begin{figure}[]
\centering
\begin{subfigure}[b]{.49\textwidth}
	\centering
	\includegraphics[width=.95\textwidth]{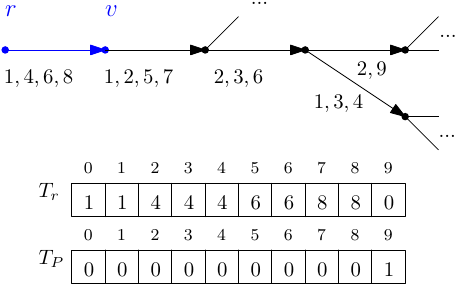}
	\caption{Starting search from root vertex $r$. 
		% Root window for $t=100$ is already removed, since no temporal path exists from $r$ to $v$ using time 100.
		\label{fig:strong_cycle_search1}}
\end{subfigure}%
\begin{subfigure}[b]{.49\textwidth}
	\centering
	\includegraphics[width=.95\textwidth]{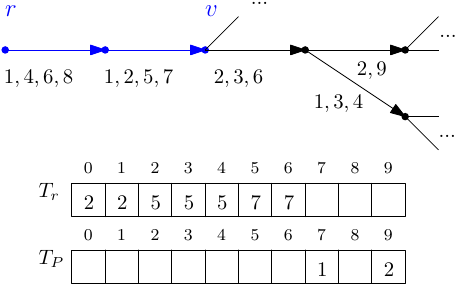}
	\caption{Search extending: update successful.  
		% No root windows are removed and path window $[0,50]$ is added.
		\label{fig:strong_cycle_search2}}
\end{subfigure}%

\begin{subfigure}[b]{.49\textwidth}
	\centering
	\includegraphics[width=.95\textwidth]{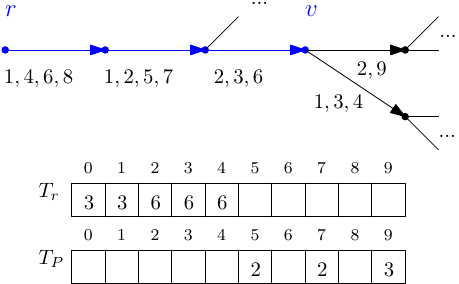}
	\caption{Another successful search extension.
		% Root window for $t=100$ is already removed, since no temporal path exists from $r$ to $v$ using time 100.
		\label{fig:strong_cycle_search3}}
\end{subfigure}%
\begin{subfigure}[b]{.49\textwidth}
	\centering
	\includegraphics[width=.95\textwidth]{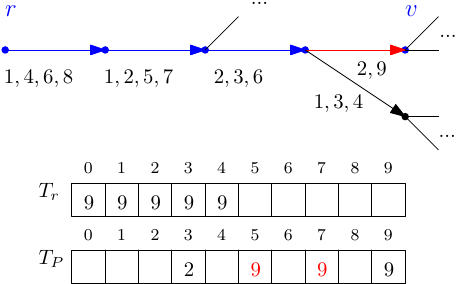}
	\caption{Update fails due to red time values.
		% No root windows are removed and path window $[0,50]$ is added.
		\label{fig:strong_cycle_search4}}
\end{subfigure}%

\begin{subfigure}[b]{.49\textwidth}
	\centering
	\includegraphics[width=.95\textwidth]{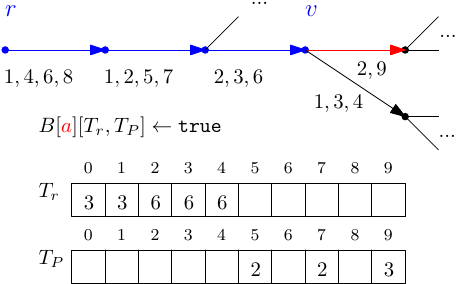}
	\caption{Backtrack and block $T_r$ and $T_P$ on $a$. 
		% Root window for $t=100$ is already removed, since no temporal path exists from $r$ to $v$ using time 100.
		\label{fig:strong_cycle_search5}}
\end{subfigure}%
\begin{subfigure}[b]{.49\textwidth}
	\centering
	\includegraphics[width=.95\textwidth]{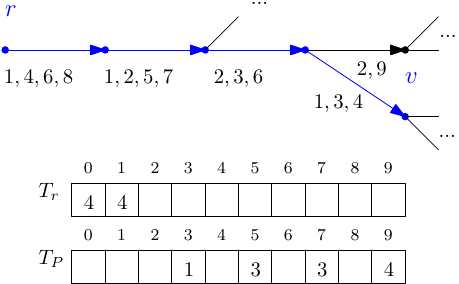}
	\caption{Search extends in different direction. 
		% No root windows are removed and path window $[0,50]$ is added.
		\label{fig:strong_cycle_search6}}
\end{subfigure}%

\caption{
	Some key operations of \Cref{algo:search_full_cycle} on a simplified example: starting and extending search, updating timetables, backtracking if update fails, and blocking timetables on arcs through the blocking matrix. The search path is shown in blue, and in red is shown an arc $a$ through which the search extends unsuccessfully (and backtracks). For clarity and simplicity in the figures, we have that: empty array cases are assumed to contain 0; the \texttt{order} functions are omitted for the blocking matrix update; and for each extension shown, we assume the arc $e$ did not block the corresponding timetables, i.e. $B[\texttt{order}(e)][\texttt{order}(T_r,T_P)] = \texttt{false}$.
		\label{fig:strong_cycle_search}}
	\end{figure}
	
%	Concerning line 8 of \Cref{algo:search_full_cycle}, the second value of $w_v$ is created by getting the maximum value among all first values of the root windows. This must correspond to the latest departure time from $r$ to $v$ because root windows $[t, \_]$ represent the reachability of temporal paths starting at time $t$ from $r$. 
%	In essence, it extends the parts of the (root or path) windows that represent the earliest arrival times to the tip of the search $v$, by taking the earliest possible times on arc %$a = (u, v)$ changed arc
%	$a = uv$. If some root window cannot extend, then it is removed. If however some path window cannot extend, then the search cannot extend through arc $a$.

%After updating the windows, they still represent the same thing as they did over the course of the search: each root window still represents some time on our initial arc, and the EAT of temporal paths using this time from $r$ to $v$, the tip of the search; and each path window still represents the EAT from some vertex $u \in P$ to $v$, the tip of the search, and the LDT from $r$ to $u$.

To prove correctness, we use the following technical lemma which is essential for proving our blocking technique doesn't hinder the detection of a strong-cycle. 
%, namely~\Cref{lemma:search_path_independence}, stated in the main text.

\searchpathindependence*

% \begin{lemma}[Search path independence]
% 	\label{lemma:search_path_independence}
% 	If a search path $P$ and another search path $Q \neq P$ both start at the same root $r$, arrive at the same vertex $u$, and have the same timetables $(T_r, T_P)$, then on any arc %$a = (u, v)$ changed arc
% 	$a = uv$ for some $v \in V \setminus (P \triangle Q)$, both search paths extend in the exact same manner, i.e. the timetables remain identical after updating. 
% \end{lemma}

%\begin{lemma}[Search path independence]
%	\label{lemma:search_path_independence}
%	If a search path $P$ and another search path $Q \neq P$ both depart from the same root $r$, arrive at the same vertex $u$, and have the same windows $(W_r, W_p)$, then on any arc %$a = (u, v)$ changed arc    $a = uv$ for some $v \in V \setminus (P \triangle Q)$, both search paths would extend in the exact same manner, i.e. the windows would remain identical after updating. 
%\end{lemma}

\begin{proof}
The values of the root timetable $T_r$ are updated depending on the times of arc $a$ only (line 5 in \Cref{algo:extend}) and so do not depend on $P$ or $Q$. Similarly, all values in $T_P$ update solely depending on the times of arc $a$ (see line 8 in \Cref{algo:extend}). Lastly, when a new value is added to $T_P$ (in lines 10 and 11 of \Cref{algo:extend}), they depend only on the root timetable $T_r$, the extended path timetable $T'_P$, which, as covered above, extends identically between $P$ and $Q$, and on the times of arc $a$. Hence, both timetables are the same between search paths $P$ and $Z$.
 \end{proof}

\begin{lemma}[Correctness]
\label{lemma:full_algo_correctness}
\Cref{algo:full_cycle_detection} returns \texttt{true} if and only if the input temporal digraph contains a strong-cycle. 
\end{lemma}

\begin{proof}
If \Cref{algo:full_cycle_detection} returns \texttt{true}, then it must come from line \ref{line:returnTrue} in \Cref{algo:search_full_cycle}, meaning some search path intersected itself. Either it has done so with the root vertex, or with a vertex of the search path. 
If it's with the root, then the corresponding full-cycle is exactly the search path and we know that each vertex can reach itself through this cycle thanks to the path timetable, which tracks exactly this property throughout the search.
If instead it intersected with a path vertex $u$, then the full-cycle is composed of $u$ and the partial search path going out from it, until it leads back to $u$. This must be a full-cycle since the non-zero values $T_P[x] = y$ in the path timetable $T_P$ indicate that for each corresponding vertex in this cycle, say $u'$, $u'$ can reach $v=u$ by time $T_P[x]$, and also, since there exists a temporal path from root $r$, through vertex $u$, to vertex $u'$, that leaves at time $x\geq T_P[x]$, there must also exist a temporal path from $u$ to $u'$ that leaves at a later time $x' \geq x$. 

Suppose the input temporal digraph contains a strong-cycle $C$, then we claim that our algorithm returns \texttt{true}. Take an arbitrary arc %$(r, v)$ changed arc
$rv$ of $C$. Suppose that our algorithm starts a search from it (the only reason it would not start such a search at some point, is if it already has returned \texttt{true} for some other reason). In the best case scenario, the search path $P$ follows exactly along the arcs of $C$. Since we know this search path corresponds to a strong-cycle, we know by definition that the vertices can all reach $r$ and then reach back to themselves in $C$, which means that the timetables of $P$ must extend without issue at every new arc exploration, as they (eventually) represent the vertices' reachability to and from $r$. Hence, in this best case scenario, the algorithm detects the strong-cycle and returns \texttt{true}. 
In any other case, since our algorithm has a blocking technique that blocks search paths depending on their timetables, any search path which has different timetables will not block the eventual exploration of this specific path (and thus the detection of the strong-cycle). When another search $Q$ with the exact same timetables explores some arc $a$ of $C$ however, then either $Q$ intersects itself and returns \texttt{true}, or it does not intersect itself and so by the search path independence lemma (\Cref{lemma:search_path_independence}), $Q$ extends in the exact same manner as $P$, meaning $Q$ will find $r$ as well and return \texttt{true}.
 \end{proof}

\begin{lemma}[Complexity]
\label{lemma:full_algo_complexity}
\Cref{algo:full_cycle_detection} runs in time $O(|\mathcal{D}|^{O(1)} \times f(\tau))$, for some function $f$, and with $\tau$ the lifetime of input temporal digraph $\mathcal{D}$. More precisely, we obtain a running time of $O(nm^2  \tau^2 4^{(\tau + 1)^2})$.
\end{lemma}

\begin{proof}
In \Cref{algo:full_cycle_detection}, line 1 runs a loop of time $O(m)$, within which at line 4 a loop of time $O(\tau)$ is run, and then a call to \Cref{algo:search_full_cycle} is made. 
In \Cref{algo:search_full_cycle}, line 2 is a loop requiring time $O(n)$. 
%Concerning line 4, since there can be at most $\tau^2$ amount of distinct windows, and thus at most $2^{\tau^2}$ amount of distinct window sets, there can be at most $4^{\tau^2}$ distinct window set pairs, meaning line 4 takes time $O(4^{\tau^2})$. 
A call to \Cref{algo:extend} is made, 
%and (in the worst case), line 8 takes time $O(\tau)$ since it goes through $W_r$, 
and (in the worst case) a recursive call to \Cref{algo:search_full_cycle} is made.
For \Cref{algo:extend}, the loop on line 3 repeats $O(\tau)$ times, and lines 4, 5, 7, and 8 all run in time $O(\tau)$ each, due to the functions called going over $\lambda(a)$. 
Line 10 goes over $T_r$, taking time $O(\tau)$, and line 11 goes over $\lambda(a)$, also taking $O(\tau)$ time.
%	\footnote{This is surely possible to improve upon; lines 5 and 6 being suboptimal on purpose for clearer analyis.} 
%Line 8 is a loop that repeats at most $O(2^{\tau^2})$ times. 
\Cref{algo:extend} thus runs in time $O(\tau^2)$.
To understand the running time of \Cref{algo:search_full_cycle} concerning the recursive calls, note that this algorithm is describing the exploration of an arc at each recursive call. However, each arc can only be explored when the search path has corresponding timetables that are not blocked by the arc, and if it backtracks, then those timetables are blocked. This implies that an arc can only be explored $O(4^{(\tau + 1)^2})$ times, and thus \Cref{algo:search_full_cycle} can only be recursively called that many times per arc as well. \Cref{algo:search_full_cycle} thus has a running time (including recursive calls) of $O(nm \tau^2 4^{(\tau + 1)^2})$. 
Inserting this into \Cref{algo:full_cycle_detection} gives a total running time of $O(nm^2  \tau^2 4^{(\tau + 1)^2})$.
 \end{proof}

%\Cref{algo:full_cycle_detection},
\Cref{lemma:full_algo_correctness} and \Cref{lemma:full_algo_complexity} together provide the proof of~\Cref{theorem:strong_cycle_detection_fpt}.

\strongcycledetectionfpt*

% \begin{theorem}
% 	\label{theorem:strong_cycle_detection_fpt}
% 	\textsc{Strong} \CycleDetection is fixed-parameter tractable with the parameter being the lifetime.
% \end{theorem} 

\section{Omitted Proofs for \AcyclicTemporization}

\subsection{Proof of~\Cref{prop:strongtemp}}
\begin{proof}
By construction, there is no cycle in timestep~1 nor in timestep~2. Now, if $C=(v_1,v_2,\ldots,v_q)$ is a cycle in $D$, suppose without loss of generality, that $\lambda(v_1v_2) =1$ and $\lambda(v_qv_1) = 2$. Then clearly $C$ does not contain any non-trivial \tpath{$v_q,v_q$}.
 \end{proof}

\subsection{Proof of~\Cref{prop:C4_label}}
\begin{proof}
First observe that if $\mathcal{D}$ has a cycle on 2 vertices, $(a,e_1,b,e_2,a)$, then whatever assignment we give to $e_1$ and $e_2$ will give us either a non-trivial temporal $a,a$-path or a non-trivial temporal $b,b$-path. Now consider that $\mathcal{D}$ has a cycle on  3 vertices, $C= (a,e_1,b,e_2,c,e_3,a)$. We can suppose, without loss of generality that $1\in \lambda(e_1)$. Indeed if this is not the case, then $C$ is contained in timestep 2, a contradiction as $\mathcal{D}$ has no {\simplecycle}s. Similarly, we can suppose that $2\in \lambda(e_3)$. As $\lambda(e_2)$ is non-empty, we get a non-trivial temporal $a,a$-path, a contradiction. 

For the second part, let $C= (a,e_1,b,e_2,c,e_3,d,e_4,a)$ be a cycle in $D$. First, we argue that $\lambda(e_1)\cap \lambda(e_2) = \emptyset$. Indeed, if it is not the case, then one can verify that in such case $(a,e_1,b,e_2,c)$ behave as a single arc and we can apply an argument analogous to the previous paragraph to arrive to a contradiction. This gives us actually that no two consecutive arcs of $C$ can be active in a single timestep of $\mathcal{D}$. As $\tau=2$ and $\lambda(e_i)\neq \emptyset$ for every $i\in [4]$, the proposition follows.
 \end{proof}

\subsection{Proof of~\Cref{thm:simpleacy}}
\begin{proof}
By Proposition~\ref{prop:simple_cycle_detection_poly}, we know that the problem is in $\NP$. Now, let $D$ be constructed as previously explained. We want to prove that $\phi$ has a NAE truth assignment if and only if $D$ admits a temporization $\lambda:E(D)\rightarrow 2^{[2]}$ such that $\tG = (D,\lambda)$ contains no {\simplecycle}s. 

First, suppose that $\phi$ has a NAE truth assignment. For each true variable $x_i$, assign time $\{1\}$ to the vertical arcs in $x_i$'s gadget and $\{2\}$ to the horizontal arcs. Do the opposite to the false variables. 
We first argue that this partial assignment does not contain {\simplecycle}s. Suppose otherwise and let $P = (v_1,t_1,v_2,\ldots,v_q,t_q,v_1)$ be a non-trivial \tpath{$v_1,v_1$} in $\tG$ (in other words, $(v_1,\ldots,v_q,v_1)$ is a {\simplecycle} in $D$). 
By construction, we know that $P$ is not contained in any $D_i$. Additionally, because each column $C_j$ forms an acyclic digraph, we get that $P$ must contain vertices of at least two distinct columns, which in turn implies that there exists $j\in [m-1]$ such that $P$ intersects $C_{2j}$. Suppose then that $i\in [n]$ and $k\in [q]$ are such that $v_k=a^i_{2j}$ and $v_{k+1}=b^i_{2j}$. This means that $v_{k-1}$ and $v_{k+2}$ are also within $D_i$, but are not in column $C_{2j}$. We then get that $t_{k-1} = t_{k+1}\neq t_{k}$. Because the rows are also acyclic, $P$ must contain another vertical arc. Whenever it happens, the times will have to alternate again, meaning that $P$  cannot be a temporal path. 

We now extend this assignment to the rest of the arcs of $D$ in a way as to not create any {\simplecycle}. So consider $j\in [m]$ and use the same notation as during the construction. If $e_1$ is assigned with time $\{1\}$, then assign to $c_jb^{i_1}_j$ time $\{2\}$ and to $a^{i_3}_jc_j$ time $\{1\}$; %otherwise, label them with $\{2\}$. 
otherwise, assign to $c_jb^{i_1}_j$ time $\{1\}$ and to $a^{i_3}_jc_j$ time $\{2\}$.
Because we have a NAE assignment, either $e_2$ or $e_3$ are given a time distinct from $e_1$ and one can see that we then do not create cycles. 

Now suppose that $\lambda:E(D)\rightarrow 2^{[2]}$ is a temporization of $D$ such that $\tG = (D,\lambda)$ contains no {\simplecycle}s. By Proposition~\ref{prop:C4_label}, we get that all vertical arcs of $D_i$ are given the same time, as well as all horizontal arcs. Assign true to $x_i$ if the vertical arcs are given time $\{1\}$; otherwise, assign false. Now suppose that $c_j = (x_{i_1}\vee x_{i_2}\vee x_{i_3})$ is such that all variables of $c_j$ receive the same value. One can verify that whatever time is given to $a^{i_3}_jc_j$ and $c_jb_j^{i_1}$, we obtain a {\simplecycle}.
 \end{proof}

\subsection{Proof of \Cref{prop:C6_label}}
\begin{proof}
Note that a {\simplecycle} is also a {\weakcycle}. Hence $D$ has no cycles on less than $4$ vertices by \Cref{prop:C4_label}. Now, consider $C= (a,e_1,b,e_2,c,e_3,d,e_4,a)$ in $D$. Proposition~\ref{prop:C4_label} also gives us that $\lambda(e_1) =\lambda(e_3)$, and $\lambda(e_2) = \lambda(e_4)$. Suppose, without loss of generality, that $\lambda(e_1) = \{1\}$. Then $C$ is a {\weakcycle} as it contains the temporal paths $(a,1,b,2,c)$ and $(c,1,d,2,a)$. If $D$ has a cycle on 5 vertices, there must be two consecutive arcs which are active at the same time. These arcs behave like a single arc in the cycle and one can apply arguments similar to ones applied to cycles on length 4 to get a contradiction. 
Finally, for the second part, consider cycle $C = (a,e_1,b,e_2,c,e_3,d,e_4,e, e_5,f,e_6,a)$ in $D$. We can suppose again that no consecutive arcs are active at the same timestep as this would be similar to the cycle on 5 vertices. The only possibility therefore is for the odd arcs to be active in one timestep, say $i$, while the even arcs are active in the other timestep $j\in[2]\setminus \{i\}$. The proposition follows.
 \end{proof}

\subsection{Proof of \Cref{thm:weakacy}}
\begin{proof}
By Proposition~\ref{prop:weak_cycle_detection_poly}, we know that the problem is in $\NP$. Now, let $D$ be constructed as previously explained. We want to prove that $\phi$ has a NAE truth assignment if and only if $D$ admits a temporization $\lambda:E(D)\rightarrow 2^{[2]}$ such that $\tG = (D,\lambda)$ contains no {\weakcycle}s. 

First, suppose that $\phi$ has a NAE truth assignment. For each true variable $x_i$, assign time $\{1\}$ to the vertical arcs in $x_i$'s gadget, assign $\{2\}$ to the arcs between vertices of $A$ and assign times $\{2\}$, $\{1\}$ and $\{2\}$ respectively for the three arcs in the path from $b^i_{2j}$ to $b^i_{2j-1}$, as well as for the three arcs in the path from $b^i_{2j}$ to $b^i_{2j+1}$. Do the opposite to the false variables. Note that, similarly to the reduction for \simplecycle, this partial assignment does not create any \weakcycle. This holds because the addition of two new vertices between consecutive vertices of $B$ forces the inclusion of one path on 3 vertices whose arcs have alternating times distinct from the times of vertical arcs. 

Now, we extend this temporization to the rest of the arcs of $D$ without creating a \weakcycle. Consider $j \in [m]$ and use the same notation as before. If $e_1$ and $e_3$ have the same time, assign times to the arc leaving $a^{i_3}_j$ and to the arc arriving at $b^{i_1}_j$ with the opposite time, assigning times to the other two arcs with the same time as $e_1$. If $e_1$ and $e_3$ have distinct times, assign time to the other four arcs in the clause gadget with alternating times in such a way that the time of the arc arriving at $b^{i_1}_j$ is distinct from $e_1$. Because we have a NAE assignment, we know that there are exactly two consecutive arcs with the same time, and thus the cycle on $7$ vertices does not induce a \weakcycle.

Now suppose that $\lambda:E(D)\rightarrow 2^{[2]}$ is a temporization of $D$ such that $\tG = (D,\lambda)$ contains no {\weakcycle}s. By Proposition~\ref{prop:C4_label}, we get that all vertical arcs of $D_i$ are given the same time. Furthermore, we know that the arcs between vertices in $A$ are given the same time as the horizontal arcs leaving and arriving at vertices in $B$, together with the fact that the path between vertices in $B$ is alternating. Assign true to $x_i$ if the vertical times are equal to $\{1\}$; otherwise, assign false. Now suppose that $c_j = (x_{i_1}\vee x_{i_2}\vee x_{i_3})$ is such that all variables of $c_j$ receive the same value. Thus the cycle on $7$ vertices corresponding to $c_j$ has $3$ consecutive arcs assigned the same time. One can verify that whatever times are given to the other $4$ arcs of the cycle, we obtain a \weakcycle.
 \end{proof}

\end{document}